\title{Polar codes for the $m$-user multiple access channels}
\author{
        Rajai Nasser
}
\date{\today}

\documentclass[11pt]{article}

\usepackage{amsthm}
\usepackage{amsfonts,amssymb,amsmath}
\usepackage{graphicx}
\usepackage{epsfig}
\usepackage[applemac]{inputenc}
\usepackage{latexsym}

\newtheorem{mydef}{Definition}
\newtheorem{mycor}{Corollary}
\newtheorem{myprop}{Proposition}
\newtheorem{mythe}{Theorem}
\newtheorem{mylem}{Lemma}
\newtheorem{mynot}{Notation}
\newtheorem{myrem}{Remark}

\newtheorem{myconj}{Conjecture}

\begin{document}


\begin{titlepage}
\vspace{4cm}
\begin{center}
\LARGE \textbf{Polar codes for the $m$-user multiple access channels\\}
\vspace{4cm}
\LARGE{School of Computer and Communication Sciences \\Ecole Polytechnique F\'{e}d\'{e}rale de Lausanne\\}
\vspace{4.5cm}
\LARGE \textbf{Semester Project \\ Spring 2010-2011}\\
\end{center}
\vspace{6cm}
\begin{tabbing}
By: \quad \quad \quad \quad \quad \= Rajai Nasser\\
\\
Supervised by:  \> Prof. Emre Telatar \\
\> Eren \c{S}a\c{s}o\u{g}lu\\
\\
Submitted on: \> \;\;10.06.2011
\end{tabbing}
\end{titlepage}

\section*{Abstract}

Polar codes are constructed for $m$-user multiple access channels (MAC) whose input alphabet size is a prime number. The block error probability under successive cancelation decoding is $o(2^{-N^{1/2-\epsilon}})$, where $N$ is the block length. Although the sum capacity is achieved by this coding scheme, some points in the symmetric capacity region may not be achieved. In the case where the channel is a combination of linear channels, we provide a necessary and sufficient condition characterizing the channels whose symmetric capacity region is preserved upon the polarization process. We also provide a sufficient condition for having a total loss in the dominant face.

\section{Introduction}

Polar coding, invented by Ar{\i}kan \cite{Arikan}, is the first low complexity coding technique that achieves the symmetric capacity of binary-input memoryless channels. The idea is to convert a set of identical copies of a given single user binary-input channel, into a set of almost extremal channels, i.e. either almost perfect channels, or almost useless channels. This phenomenon is called \emph{polarization}.\\

Ar{\i}kan's technique was generalized in \cite{SasogluTelAri} for channels with arbitrary input alphabet size. The probability of error of successive cancelation decoding of polar codes was proved to be equal to $o(2^{-N^{1/2-\epsilon}})$ \cite{ArikanTelatar}.\\

In the case of multiple access channels, we find two main results in the literature: (i) E. \c{S}a\c{s}o\u{g}lu et al. constructed polar codes for the two-user MAC \cite{SasogluTelYeh}, (ii) E. Abbe and E. Telatar constructed polar codes for the $m$-user MAC with binary input \cite{AbbeTelatar}. In this project, we combine the ideas of \cite{SasogluTelYeh} and \cite{AbbeTelatar} to construct polar codes for the $m$-user MAC whose input alphabet size is a prime number (or a power of a prime number).\\

In our construction, as well as in both constructions in \cite{SasogluTelYeh} and \cite{AbbeTelatar}, the sum capacity is preserved upon the polarization process but some points in the symmetric capacity region are not always achieved by the polar coding scheme. A part of the symmetric capacity region is lost by polar coding. In this project, we study this loss in the special case where the channel is a combination of linear channels (this class of channels will be introduced in section 6).\\

In section 2, we introduce the preliminaries for this project. We describe the polarization process in section 3. The rate of polarization is studied in section 4. Polar codes for the $m$-user MAC are constructed in section 5. The problem of loss in the capacity region is studied in section 6. Finally, we conclude this project in section 7.

\section{Preliminaries}

\begin{mydef}
A discrete $m$-user multiple access channel (MAC) is an ($m+2$)-tuple $P=(\mathcal{X}_1,\;\mathcal{X}_2,\;...,\;\mathcal{X}_m,\;\mathcal{Y},\;f_P)$ where $\mathcal{X}_1,\;...,\;\mathcal{X}_m$ are finite sets that are called the ``input alphabets'' of $P$, $\mathcal{Y}$ is a finite set that is called the ``output alphabet'' of $P$, and $f_P:\mathcal{X}_1\times\mathcal{X}_2\times...\times\mathcal{X}_m \times \mathcal{Y} \rightarrow [0,1]$  is a function satisfying $\forall(x_1,x_2,...,x_m)\in\mathcal{X}_1\times\mathcal{X}_2\times...\times\mathcal{X}_m,\;\displaystyle\sum_{y\in\mathcal{Y}} f_P(x_1,x_2,...,x_m,y)=1$.
\end{mydef}

\begin{mynot}
We write $P: \mathcal{X}_1\times\mathcal{X}_2\times...\times\mathcal{X}_m \rightarrow \mathcal{Y}$ to denote that $P$ has $m$ users, $\mathcal{X}_1,\;\mathcal{X}_2,\;...,\;\mathcal{X}_m$ as input alphabets, and $\mathcal{Y}$ as output alphabet. We denote $f_P(x_1,x_2,...,x_m,y)$ by $P(y|x_1,x_2,...,x_m)$ which is interpreted as the conditional probability of getting $y$ at the output, given that we have $(x_1,x_2,...,x_m)$ at the input.
\end{mynot}

\begin{mydef}
A code $\mathcal{C}$ of block length $N$ and rate vector $(R_1,R_2,...,R_m)$ is an $(m+1)$-tuple $\mathcal{C}=(f_1,f_2,...,f_m,g)$, where $f_k: \mathcal{W}_k=\{1,2,...,\lceil\alpha^{NR_k}\rceil\}\rightarrow\mathcal{X}_k^N$ is the encoding function of the $k^{th}$ user and $g:\mathcal{Y}^n\rightarrow\mathcal{W}_1\times\mathcal{W}_2\times...\times\mathcal{W}_m$ is the decoding function. We denote $f_k(w)=\big(f_k(w)_1,...,f_k(w)_N\big)$, where $f_k(w)_n$ is the $n^{th}$ component of $f_k(w)$. The average probability of error of the code $\mathcal{C}$ is given by:
\begin{align*}
P_e(\mathcal{C})&=\sum_{(w_1,...,w_m)\in\mathcal{W}_1\times...\times\mathcal{W}_m}
\frac{P_e(w_1,...,w_m)}{|\mathcal{W}_1|\times...\times|\mathcal{W}_m|}
\\
P_e(w_1,...,w_m)&=\sum_{\substack{(y_1,...,y_N)\in\mathcal{Y}^N\\g(y_1,...,y_N)\neq(w_1,...,w_m)}}\prod_{n=1}^N P\big(y_n|f_1(w_1)_n,...,f_m(w_m)_n\big)
\end{align*}
where $\alpha$ is a pre-determined real number according to which the rate of information is measured. If $\alpha=2$, the rates are expressed in bits, and if $\alpha=e$, the rates are expressed in nats.
\end{mydef}

\begin{mydef}
A rate vector $R=(R_1,...,R_m)$ is said to be achievable if there exists a sequence of codes $\mathcal{C}_N$ of rate vector $R$ and of block length $N$ such that $P_e(\mathcal{C}_N)$ tends to zeros as $N$ tends to infinity. The capacity region of the MAC $P$ is the closure of the set of all achievable rate vectors.
\end{mydef}

\begin{mythe}
(Theorem 15.3.6 \cite{Cover}) The capacity region of a MAC $P$ is given by the closure of the convex hull of the rate vectors $(R_1,...,R_m)$ satisfying
$$R(S) \leq I[S](P)\;\;for\;all\;\;S\subseteq\{1,...,m\}$$
for some probability distribution $P(y|x_1,...,x_m)p_1(x_1)...p_m(x_m)$ on $\mathcal{X}_1\times...\times\mathcal{X}_m\times\mathcal{Y}$.
where $R(S):=\displaystyle\sum_{k=1}^{l_S} R_k$, $X(S):=(X_{s_1},...,X_{s_{l_S}})$ for $S=\{s_1,...,s_{l_S}\}$ and $I[S](P):=I(X(S);YX(S^c))$. All the mutual informations are calculated using logarithm base $\alpha$.
\end{mythe}

\begin{mydef}
$I(P):=I[\{1,...,m\}](P)$ is called the \emph{sum capacity} of $P$, it's equal to the maximum value of $R_1+...+R_m$ when $(R_1,...,R_m)$ is achievable. The set of points of the capacity region satisfying $R_1+...+R_m=I[\{1,...,m\}](P)$ is called \emph{the dominant face}.
\end{mydef}

In this project, we are interested in MAC where $\mathcal{X}_k=\mathbb{F}_q$ and we take $\alpha=q$. We will focus our attention on the case where $q$ is a prime number since we can easily generalize for the more general case where $q$ is a power of a prime number. More particularly, we are interested in the \emph{symmetric capacity region} which is defined by:
$$\mathcal{J}(P)=\{(R_1,...,R_m):\:0\leq R(S)\leq I(X(S);YX(S^c))\;\forall S\subseteq\{1,...,m\}\}$$
where the mutual informations are calculated for independently, uniformly distributed $X_1$,...,$X_m$.

\section{Polarization process}

\begin{mydef}
Let $P:\mathbb{F}_q^m\rightarrow\mathcal{Y}$ be a discrete $m$-user MAC. We define the two channels $P^-:\mathbb{F}_q^m \rightarrow \mathcal{Y}^2$ and $P^+:\mathbb{F}_q^m \rightarrow \mathcal{Y}^2\times\mathbb{F}_q^m$ as:
$$P^-(y_1,y_2|u_1^1,...,u_m^1)=\sum_{(u_1^2,...,u_m^2)\in\mathbb{F}_q^m}
\frac{1}{q^m}P(y_1|u_1^1+u_1^2,...,u_m^1+u_m^2)P(y_2|u_{1}^2,...,u_m^2)$$
$$P^+(y_1,y_2,u_1^1,...,u_m^1|u_1^2,...,u_m^2)=\frac{1}{q^m}P(y_1|u_1^1+u_1^2,...,u_m^1+u_m^2)P(y_2|u_{1}^2,...,u_m^2)$$
\end{mydef}

$P^-$ and $P^+$ can be constructed from two independent copies of $P$ as follows: The $k^{th}$ user chooses independently and uniformly two symbols $U_k^1$ and $U_k^2$ in $\mathbb{F}_q$, he calculates $X^1_k=U_k^1+U_k^2$ and $X^2_k=U_k^2$, and he finally sends $X^1_k$ through the first copy of $P$ and $X^2_k$ through the second copy of $P$. $P^-$ and $P^+$ are the channels $U_1^1...U_m^1\rightarrow Y_1Y_2$ and $U_1^2...U_m^2\rightarrow Y_1Y_2U_1^1...U_m^1$ respectively, where $Y_1$ and $Y_2$ are the respective outputs of the first and second copy of $P$.\\

Note that the transformation $(U_1^1,...,U_m^1,U_1^2,...,U_m^2)\rightarrow(X_1^1,...,X_m^1,X_1^2,...,X_m^2)$ is bijective and therefore it induces uniform and independent distributions for $X_1^1,...,X_m^1,X_1^2,...,X_m^2$ which are the inputs of the $P$ channels.

\begin{mydef}
Let $\{B_n\}_{n\geq1}$ be i.i.d. uniform random variables on $\{-,+\}$. We define the $MAC$-valued process $\{P_n\}_{n\geq0}$ by:
\begin{align*}
P_0 &:= P\\
P_{n} &:=P_{n-1}^{B_n}\;\forall n\geq1
\end{align*}
\end{mydef}

\begin{myprop}
The process $\{I[S](P_n)\}_{n\geq0}$ is a bounded super-martingale for all $S\subset\{1,...,m\}$. Moreover, it's a bounded martingale if $S=\{1,...,m\}$.
\end{myprop}
\begin{proof}
\begin{align*}
2I[S](P)&=I[S](P)+I[S](P)=I(X^1(S);Y_1X^1(S^c))+I(X^2(S);Y_2X^2(S^c))\\
&= I(X^1(S)X^2(S);Y_1Y_2X^1(S^c)X^2(S^c))=I(U^1(S)U^2(S);Y_1Y_2U^1(S^c)U^2(S^c))\\
&=I(U^1(S);Y_1Y_2U^1(S^c)U^2(S^c))+I(U^2(S);Y_1Y_2U^1(S^c)U^2(S^c)U^1(S))\\
&\geq I(U^1(S);Y_1Y_2U^1(S^c)) + I(U^2(S);Y_1Y_2U_1^1...U_m^1U^2(S^c))\\
&= I[S](P^-) + I[S](P^+).
\end{align*}
Thus, $E\big(I[S](P_{n+1})\big|P_{n}\big)=\frac{1}{2}I[S](P_{n}^-)+\frac{1}{2}I[S](P_{n}^+)\leq I[S](P_n)$, and $I[S](P_n)\leq |S|$ for all $S\subset\{1,...,m\}$, which proves that $\{I[S](P_n)\}_{n\geq0}$ is a bounded super-martingale. If $S=\{1,...,m\}$, the inequality becomes equality, and $\{I[S](P_n)\}_{n\geq0}$ is a bounded martingale.
\end{proof}

Since $\frac{1}{2}(I[S](P^-)+I[S](P^+))\leq I[S](P)$ $\forall S\subset \{1,...,m\}$, then $\frac{1}{2}\mathcal{J}(P-)+\frac{1}{2}\mathcal{J}(P+)\subset\mathcal{J}(P)$, but this subset relation can be strict if one of the inequalities is strict for a certain $S\subset\{1,...,m\}$. Nevertheless, for $S=\{1,...,m\}$, we have $\frac{1}{2}(I(P^-)+I(P^+))= I(P)$, so at least one point of the dominant face of $\mathcal{J}(P)$ is present in $\frac{1}{2}\mathcal{J}(P-)+\frac{1}{2}\mathcal{J}(P+)$ since the capacity region is a polymatroid. Therefore, the sum capacity is preserved.\\

From the bounded super-martingale convergence theorem, we deduce that the sequences $\{I[S](P_n)\}_{n\geq0}$ converge almost surely for all $S\subset\{1,...,m\}$.\\

The main result of this section is that, almost surely, $P_n$ becomes an almost deterministic linear channel:

\begin{mythe}
Let $P$ be an $m$-user MAC. Then for every $\epsilon>0$, we have:
\begin{align*}
\lim_{l\to\infty} \frac{1}{2^l}\Big|\big\{ & s\in\{-,+\}^l: \exists A_{s}\in\mathbb{F}_q^{m\times r_s}, \emph{\textrm{rank}}(A_{s})=r_s,\\
& |I(A_{s}^T\vec{U}_{s};Y_{s})-I(P^s)|<\epsilon, |r_s-I(P^s)|<\epsilon \big\}\Big| = 1
\end{align*}
$\mathbb{F}_q^{m\times r_s}$ denotes the set of $m\times r_s$ matrices with coefficients in $\mathbb{F}_q$, $U_s$ (uniform random vector in $\mathbb{F}_q^m$) is the input to the channel $P^s$, and $Y_s$ is the output of it.
\end{mythe}

To prove this theorem, we need several lemmas and definitions:

\begin{mylem}(\cite{Sasoglu})
For any $\epsilon>0$, there exists $\delta>0$ such that for any single user channel $P$ we have:
$$I(P^+)-I(P)<\delta \Rightarrow I(P)\notin (\epsilon,1-\epsilon)$$
\end{mylem}

\begin{mydef}
Let $P:\mathbb{F}_q^m\rightarrow \mathcal{Y}$ be an $m$-user MAC, let $A\in \mathbb{F}_q^{m\times n_1}$, $B\in \mathbb{F}_q^{m\times n_2}$ be two matrices such that $[A\;B]\in\mathbb{F}_q^{m\times (n_1+n_2)}$ has rank $n_1 + n_2$ (If $n_2=0$, we put $B=\Phi$). We define the $n_1$-user channel $P[A|B]:\mathbb{F}_q^{n_1}\rightarrow \mathcal{Y}\times\mathbb{F}_q^{n_2}$ as follows:
$$P[A|B](y,\vec{v}|\vec{u})=\frac{1}{q^{m-n_1}}\sum_{\substack{\vec{x}\in\mathbb{F}_q^m\\ A^T\vec{x}=\vec{u},B^T\vec{x}=\vec{v}}}P(y|\vec{x})$$
If $n_2=0$, there is no additional $\vec{v}$ at the output, and the $B^T\vec{x}=\vec{v}$ constraint under the sum is removed.
\end{mydef}

\begin{myrem}
The channel $P[A|B]$ can be constructed from $P$ as follows: Let $\vec{x}\in \mathbb{F}_q^m$ be the input to the channel $P$, $\vec{u}=A^T\vec{x}$ is determined by the $n_1$ users of $P[A|B]$. Since $n_1$ may be less than $m$, then $\vec{x}$ cannot always be determined from $\vec{u}$. $\vec{u}$ only determines a part of $\vec{x}$, the other part is determined by the channel $P[A|B]$ uniformly and independently from $\vec{u}$:\\

Let $\tilde{A}$ be an $m\times(m-n_1)$ matrix such that $[A\;\tilde{A}]$ is invertible, i.e. the columns of $\tilde{A}$ together with those of $A$ form a basis for $\mathbb{F}_q^m$. Since $[A\;\tilde{A}]$ is invertible, we could determine $\vec{x}$ if we knew $[A\;\tilde{A}]^T\vec{x}$. But $[A\;\tilde{A}]^T\vec{x}=[(A^T\vec{x})^T\;(\tilde{A}^T\vec{x})^T]^T=[\vec{u}^T\;\vec{\tilde{u}}^T]^T$. We only have $\vec{u}$ and we need $\vec{\tilde{u}}$ to fully determine $\vec{x}$. The channel $P[A|B]$ generates $\vec{\tilde{u}}$ uniformly and independently from $\vec{u}$. The output of the channel $P[A|B]$ is $y$ (the output of $P$) together with $B^T\vec{x}$. \textbf{In other words, in $P[A|B]$, we try to determine $A^T\vec{x}$, when $\vec{y}$ (the output of $P$) and $B^T\vec{x}$ are given}.\\

In summary, $P[A|B]$ corresponds to the channel $\vec{U}=A^T\vec{X}\rightarrow (Y,B^T\vec{X})$, where $\vec{U}$ is a uniform random vector in $\mathbb{F}_q^{n_1}$, $\vec{X}=([A\;\tilde{A}]^T)^{-1}[\vec{U}^T\;\vec{\tilde{U}}^T]^T$ and $Y$ is the output of the channel $P$ when $\vec{X}$ is the input. $\tilde{A}$ is any $m\times(m-n_1)$ matrix such that $[A\;\tilde{A}]$ is invertible, and $\vec{\tilde{U}}$ is a uniform random vector in $\mathbb{F}_q^{m-n_1}$, independent from $\vec{U}$. We have $I(P[A|B])=I(\vec{U};Y,B^T\vec{X})=I(A^T\vec{X};Y,B^T\vec{X})$.\\

Let $A=[A'\;A'']$ and $B$ be two matrices having $m$ rows such that $[A\;B]$ is full rank. If $X$ and $Y$ are as above, then we have:
\begin{align*}
I(P[A|B])&=I(A^T\vec{X};Y,B^T\vec{X})=I(A'^T\vec{X},A''^T\vec{X};Y,B^T\vec{X})\\
&=I(A'^T\vec{X};Y,B^T\vec{X})+I(A''^T\vec{X};Y,B^T\vec{X},A'^T\vec{X})\\
&=I(A'^T\vec{X};Y,B^T\vec{X})+I(A''^T\vec{X};Y,[B\;A']^T\vec{X})\\
&=I(P[A'|B])+I\big(P\big[A''\big|[B\;A']\big]\big)
\end{align*}
\end{myrem}

\begin{mylem}
Let $A\in \mathbb{F}_q^{m\times n_1}$, $B\in \mathbb{F}_q^{m\times n_2}$, $A'\in \mathbb{F}_q^{n_1\times n_1'}$, $B'\in \mathbb{F}_q^{n_1\times n_2'}$ be four matrices such that $[A\;B]\in\mathbb{F}_q^{m\times (n_1+n_2)}$ has rank $n_1 + n_2$, and $[A'\;B']\in\mathbb{F}_q^{n_1\times (n_1'+n_2')}$ has rank $n_1' + n_2'$. Then $P[A|B][A'|B']$ is equivalent to $P\big[AA'\big|[B\;AB']\big]$
\end{mylem}
\begin{proof}
$P[A|B][A'|B'](y,\vec{b},\vec{b'}|\vec{a'})$
\begin{align*}
&=\frac{1}{q^{n_1-n_1'}}\sum_{\substack{\vec{a}\in\mathbb{F}_q^{n_1}\\ A'^T\vec{a}=\vec{a}'\\B'^T\vec{a}=\vec{b}'}}P[A|B](y,\vec{b}|\vec{a})
=\frac{1}{q^{n_1-n_1'}}\sum_{\substack{\vec{a}\in\mathbb{F}_q^{n_1}\\ A'^T\vec{a}=\vec{a}'\\B'^T\vec{a}=\vec{b}'}}
\frac{1}{q^{m-n_1}}\sum_{\substack{\vec{u}\in\mathbb{F}_q^m\\ A^T\vec{u}=\vec{a}\\B^T\vec{u}=\vec{b}}}P(y|\vec{u})\\
&=\frac{1}{q^{m-n_1'}}
\sum_{\substack{\vec{u}\in\mathbb{F}_q^m\\ A'^TA^T\vec{u}=\vec{a}'\\B'^TA^T\vec{u}=\vec{b}',B^T\vec{u}=\vec{b}}}P(y|\vec{u})
=\frac{1}{q^{m-n_1'}}
\sum_{\substack{\vec{u}\in\mathbb{F}_q^m\\ (AA')^T\vec{u}=\vec{a}'\\ [B\; AB']^T\vec{u}=[\vec{b}^T\;\vec{b}'^T]^T}}P(y|\vec{u})\\
&=P\big[AA'\big|[B\;AB']\big]\big(y,[\vec{b}^T\;\vec{b}'^T]^T\big|\vec{a}'\big)
\end{align*}
\end{proof}

\begin{mylem}
$P[A|B]^+$ is degraded with respect to $P^+[A|B]$. $P^-[A|B]$ is degraded with respect to $P[A|B]^-$, and if $B=\Phi$, they are equivalent.
\end{mylem}
\begin{proof}
\begin{align*}
P^+[A|B](y_1,y_2,\vec{u}_1,\vec{b}|\vec{a})&=\frac{1}{q^{m-n_1}}\sum_{\substack{\vec{u}_2\in\mathbb{F}_q^m\\A^T\vec{u}_2=\vec{a},B^T\vec{u}_2=
\vec{b}}}P^+(y_1,y_2,\vec{u}_1|\vec{u}_2)\\
&=\frac{1}{q^{2m-n_1}}\sum_{\substack{\vec{u}_2\in\mathbb{F}_q^m\\A^T\vec{u}_2=\vec{a},B^T\vec{u}_2=\vec{b}}}P(y_1|\vec{u}_1+\vec{u}_2)P(y_2|\vec{u}_2)
\end{align*}

\begin{align*}
P[A|B]^+(y_1,\vec{b}_1,y_2,\vec{b}_2,\vec{a}_1|\vec{a}_2)&=\frac{1}{q^{n_1}}P[A|B](y_1,\vec{b}_1|\vec{a}_1+\vec{a}_2)P[A|B](y_2,\vec{b}_2|\vec{a}_2)\\
&=\frac{1}{q^{2m-n_1}}\sum_{\substack{\vec{u_1}\in\mathbb{F}_q^m\\ A^T\vec{u}_1=\vec{a}_1+\vec{a}_2\\B^T\vec{u}_1=\vec{b}_1}}
\sum_{\substack{\vec{u}_2\in\mathbb{F}_q^m\\ A^T\vec{u}_2=\vec{a}_2\\B^T\vec{u}_2=\vec{b}_2}}P(y_1|\vec{u}_1)P(y_2|\vec{u}_2)\\
&=\frac{1}{q^{2m-n_1}}\sum_{\substack{\vec{u_1}\in\mathbb{F}_q^m\\ A^T\vec{u}_1=\vec{a}_1\\B^T\vec{u}_1=\vec{b}_1-\vec{b}_2}}
\sum_{\substack{\vec{u}_2\in\mathbb{F}_q^m\\ A^T\vec{u}_2=\vec{a}_2\\B^T\vec{u}_2=\vec{b}_2}}P(y_1|\vec{u}_1+\vec{u}_2)P(y_2|\vec{u}_2)\\
&= \sum_{\substack{\vec{u_1}\in\mathbb{F}_q^m\\ A^T\vec{u}_1=\vec{a}_1\\B^T\vec{u}_1=\vec{b}_1-\vec{b}_2}} P^+[A|B](y_1,y_2,\vec{u}_1,\vec{b}_2|\vec{a}_2)
\end{align*}

Which proves that $P[A|B]^+$ is degraded with respect to $P^+[A|B]$.

\begin{align*}
P[A|B]^-(y_1,\vec{b}_1,y_2,\vec{b}_2|\vec{a}_1)&=\frac{1}{q^{n_1}}\sum_{\vec{a}_2\in\mathbb{F}_q^{n_1}}P[A|B](y_1,\vec{b}_1|\vec{a}_1+\vec{a}_2)
P[A|B](y_2,\vec{b}_2|\vec{a}_2)\\
&=\frac{1}{q^{2m-n_1}}\sum_{\vec{a}_2\in\mathbb{F}_q^{n_1}}\sum_{\substack{\vec{u_1}\in\mathbb{F}_q^m\\ A^T\vec{u}_1=\vec{a}_1+\vec{a}_2\\B^T\vec{u}_1=\vec{b}_1}}
\sum_{\substack{\vec{u}_2\in\mathbb{F}_q^m\\ A^T\vec{u}_2=\vec{a}_2\\B^T\vec{u}_2=\vec{b}_2}}P(y_1|\vec{u}_1)P(y_2|\vec{u}_2)\\
&=\frac{1}{q^{2m-n_1}}\sum_{\vec{a}_2\in\mathbb{F}_q^{n_1}}\sum_{\substack{\vec{u_1}\in\mathbb{F}_q^m\\ A^T\vec{u}_1=\vec{a}_1\\B^T\vec{u}_1=\vec{b}_1-\vec{b}_2}}
\sum_{\substack{\vec{u}_2\in\mathbb{F}_q^m\\ A^T\vec{u}_2=\vec{a}_2\\B^T\vec{u}_2=\vec{b}_2}}P(y_1|\vec{u}_1+\vec{u}_2)P(y_2|\vec{u}_2)\\
&= \frac{1}{q^{2m-n_1}}\sum_{\substack{\vec{u_1}\in\mathbb{F}_q^m\\ A^T\vec{u}_1=\vec{a}_1\\B^T\vec{u}_1=\vec{b}_1-\vec{b}_2}}
\sum_{\substack{\vec{u}_2\in\mathbb{F}_q^m\\B^T\vec{u}_2=\vec{b}_2}}P(y_1|\vec{u}_1+\vec{u}_2)P(y_2|\vec{u}_2)
\end{align*}

\begin{align*}
P^-[A|B](y_1,y_2,\vec{b}|\vec{a})&=\frac{1}{q^{m-n_1}}\sum_{\substack{\vec{u}_1\in\mathbb{F}_q^m\\A^T\vec{u}_1=\vec{a},B^T\vec{u}_1=
\vec{b}}}P^-(y_1,y_2|\vec{u}_1)\\
&=\frac{1}{q^{2m-n_1}}\sum_{\substack{\vec{u}_1\in\mathbb{F}_q^m\\A^T\vec{u}_1=\vec{a},B^T\vec{u}_1=\vec{b}}}\sum_{\vec{u}_2\in\mathbb{F}_q^m}
P(y_1|\vec{u}_1+\vec{u}_2)P(y_2|\vec{u}_2)\\
&=\sum_{\vec{b}_2\in\mathbb{F}_q^{n_2}}P[A|B]^-(y_1,\vec{b}+\vec{b}_2,y_2,\vec{b}_2|\vec{a})
\end{align*}

Which proves that $P^-[A|B]$ is degraded with respect to $P[A|B]^-$. Note that if $B=\Phi$, all the $B^T\vec{u}$ constraints are eliminated, and there is no more $\vec{b}$ vectors. In this case we have $P[A|\Phi]^-(y_1,y_2|\vec{a}_1)=P^-[A|\Phi](y_1,y_2|\vec{a}_1)$. Therefore, $P[A|\Phi]^-$ and $P^-[A|\Phi]$ are equivalent.
\end{proof}

\begin{myprop}
The process $I(P_n[A|B])_{n\geq}$ converges almost surely for any two full rank matrices $A$ and $B$ ($B$ can be $\Phi$) whose columns are linearly independent. Moreover, the limit takes its value in the set of integers.
\end{myprop}
\begin{proof}
Let's start with the case $B=\Phi$, we have: $P_n^-[A|\Phi]$ is equivalent to $P_n[A|\Phi]^-$ so $I(P_n^-[A|\Phi])=I(P_n[A|\Phi]^-)$, and $P_n[A|\Phi]^+$ is degraded with respect to $P_n^+[A|\Phi]$ so $I(P_n^+[A|\Phi])\geq I(P_n[A|\Phi]^+)$.\\

Let $\vec{U}'_1$, $\vec{U}'_2$ be two independent uniform random vectors in $\mathbb{F}_q^{n_1}$ which will be the input to $P_n[A|\Phi]^-$ and $P_n[A|\Phi]^+$ respectively. $\vec{U}_1=\vec{U}'_1+\vec{U}'_2$ and $\vec{U}_2=\vec{U}'_2$ will be the inputs to two independent copies of $P_n$ respectively. Let $\vec{\tilde{U}}_1$ and $\vec{\tilde{U}}_2$ be two independent uniform random vectors in $\mathbb{F}_q^{m-n_1}$ which are independent from $\vec{U}_1$ and $\vec{U}_2$, let $\vec{X}_i=([A\;\tilde{A}]^T)^{-1}[\vec{U}_i^T \;\vec{\tilde{U}}_i^T]^T$ for $i\in\{1,2\}$ (see \emph{remark 1}), and let $Y_1$ (resp. $Y_2$) be the output of the channel $P_n$ when $\vec{X}_1$ (resp. $\vec{X}_2$) is the input. $\tilde{A}$ is any $m\times(m-n_1)$ matrix such that $[A\;\tilde{A}]$ is invertible. We have: $I(P_n[A|\Phi]^-)=I(\vec{U}'_1;Y_1Y_2)$, $I(P_n[A|\Phi]^+)=I(\vec{U}'_2;Y_1Y_2\vec{U}'_1)$, $\vec{U}_k=A^T\vec{X}_k$ and $I(A^T\vec{X}_k;Y_k)=I(P_n[A|\Phi])$ for $k\in\{1,2\}$ (see \emph{remark 1}). Thus:

\begin{align*}
I(P_n^-[A|\Phi])+I(P_n^+[A|\Phi])&\geq I(P_n[A|\Phi]^-)+I(P_n[A|\Phi]^+)= I(\vec{U}'_1;Y_1Y_2) + I(\vec{U}'_2;Y_1Y_2\vec{U}'_1)\\
&=I(\vec{U}'_1\vec{U}'_2;Y_1Y_2)=I(\vec{U}_1\vec{U}_2;Y_1Y_2)=I(A^T\vec{X}_1,A^T\vec{X}_2;Y_1Y_2)\\
&=I(A^T\vec{X}_1;Y_1) + I(A^T\vec{X}_2;Y_2)=2I(P_n[A|\Phi])
\end{align*}

Therefore $E(I(P_{n+1}[A|\Phi])|P_n)=\frac{1}{2}I(P_n^-[A|\Phi])+\frac{1}{2}I(P_n^+[A|\Phi])\geq I(P_n[A|\Phi]$. Moreover we have $I(P_n[A|\Phi])\leq n_1$, so the process $I(P_n[A|\Phi])_{n\geq0}$ is a bounded sub-martingale. By the bounded sub-martingale convergence theorem we deduce the almost sure convergence of $I(P_n[A|\Phi])_{n\geq0}$ for any full rank matrix $A$.\\

From \emph{remark 1} we have $I(P_n[A|B])=I(P_n[[A\;B]|\Phi])-I(P_n[B|\Phi])$, so $I(P_n[A|B])_{n\geq0}$ also converges almost surely for any two matrices $A$ and $B$ such that $[A\;B]$ is full rank.\\

Now suppose that $A=\vec{\alpha}$ is a column vector (i.e. $n_1=1$), the almost sure convergence of $I(P_n[\vec{\alpha}|B])_{n\geq0}$ implies the almost sure convergence of $\big|I(P_{n+1}[\vec{\alpha}|B])-I(P_n[\vec{\alpha}|B])\big|$ to zero, and so $E\Big(\big|I(P_{n+1}[\vec{\alpha}|B])-I(P_n[\vec{\alpha}|B])\big|\;\Big|P_n\Big)$ converges almost surely to zero, since $I(P_n[\vec{\alpha}|B])$ is bounded. We have also:
\begin{align*}
E\Big(\big|I(P_{n+1}[\vec{\alpha}|B])-I(P_n[\vec{\alpha}|B])\big|\;\Big|P_n\Big)
&\geq\frac{1}{2}\big(I(P_n^+[\vec{\alpha}|B])-I(P_n[\vec{\alpha}|B])\big)\\
&\geq\frac{1}{2}\big(I(P_n[\vec{\alpha}|B]^+)-I(P_n[\vec{\alpha}|B])\big)
\end{align*}
The first inequality comes from the expression of the expectation and the second one comes from the fact that $I(P_n[\vec{\alpha}|B]^+)$ is degraded with respect to $I(P_n^+[\vec{\alpha}|B])$ (\emph{lemma 3}), we conclude that $I(P_n[\vec{\alpha}|B]^+)-I(P_n[\vec{\alpha}|B])$ converges almost surely to zero.\\

Let $\{P_n\}_{n\geq0}$ be a realization in which $I(P_n[\vec{\alpha}|B])_{n\geq0}$ converges to a certain value $l\in [0,1]$ (we have $0\leq I(P_n[\vec{\alpha}|B])\leq 1$). Due to the convergence of $I(P_n^+[\vec{\alpha}|B])-I(P_n[\vec{\alpha}|B])$ to zero and of $I(P_n[\vec{\alpha}|B])$ to $l$, for every $\epsilon>0$, there exists $n_0>0$ such that for any $n>n_0$ we have $|I(P_n^+[\vec{\alpha}|B])-I(P_n[\vec{\alpha}|B])|<\delta$ where $\delta$ is as in \emph{lemma 1}, and $|I(P_n[\vec{\alpha}|B])-l|<\epsilon$. We conclude from \emph{lemma 1} that $I(P_n[\vec{\alpha}|B])\notin(\epsilon,1-\epsilon)$ and since $|I(P_n[\vec{\alpha}|B])-l|<\epsilon$ then $l\notin(2\epsilon,1-2\epsilon)$, and this is true for any $\epsilon>0$. We conclude that $l\in\{0,1\}$.\\

Now let $A=[\alpha_1\;...\;\alpha_{n_1}]$ be any full rank matrix, then by \emph{remark 1} we have:
\begin{align*}
I(P_n[A|B])=\sum_{k=1}^{n_1}I\big(P_n\big[\alpha_k|[B\;\alpha_1\;...\;\alpha_{k-1}]\big]\big)
\end{align*}

Since each of $I(P_n[\alpha_k|[B\;\alpha_1\;...\;\alpha_{k-1}]])$ converges almost surely to a value in $\{0,1\}$, then $I(P_n[A|B])$ converges almost surely to an integer.
\end{proof}

\begin{mycor}
The limit of the process $I[S](P_n)$ is almost surely an integer for all $S\subset\{1,...,m\}$.
\end{mycor}
\begin{proof}
If we take $A_S=[e_k,\;k\in S]$ and $B_S=[e_k,\;k\in S^c]$, where $\{e_k\}$ is the canonical basis of $\mathbb{F}_q^m$, then $I[S](P_n)=I(P_n[A_S|B_S])$. The assertion about the limit comes from the previous proposition.
\end{proof}

\begin{mylem}(lemma 33\cite{SasogluTelYeh})
Let $X,W$ be two independent and uniformly distributed random variables in $\mathbb{F}_q$, and let $Y$ be an arbitrary random variable. For every $\epsilon'>0$, there exists $\delta>0$ such that:
\begin{itemize}
  \item $I(X,Y)<\delta$, $I(W;Y)<\delta$, $H(X|YW)<\delta$, $H(W|YX)<\delta$ and
  \item $H(\beta X+\gamma W|Y)\notin(\delta,1-\delta)$ for all $\beta,\gamma\in\mathbb{F}_q$,
\end{itemize}
implies $$I(\beta' X+ \gamma' W;Y)>1-\epsilon'$$ for some $\beta',\gamma'\in\mathbb{F}_q$
\end{mylem}

\begin{mylem}
For every $m>0$, there exists $\epsilon_m>0$ such that for any $\epsilon<\epsilon_m$, if $P$ is an $m$-user MAC satisfying $d(I(P[A|B]),\mathbb{Z})<\epsilon$ for all matrices $A$ and $B$ ($B$ can be $\Phi$), then $I(P)>1-\epsilon$ implies the existence of a non-zero vector $\vec{\alpha}\in\mathbb{F}_q^m$ satisfying $I(P[\vec{\alpha}|\Phi])>1-\epsilon$.
\end{mylem}
\begin{proof}

Choose $\delta$ as in the above lemma for $\epsilon'=\frac{1}{3}$, and then choose $\epsilon_m=\min\{\delta,\frac{1}{3},\frac{1}{m+1}\}$ (note that our choice of $\epsilon_m$ is non-increasing with $m$). Let $\{\vec{e}_k,\;1\leq k\leq m\}$ be the canonical basis of $\mathbb{F}_q^m$, let $0<\epsilon<\epsilon_m$ and let $P$ be an $m$-user MAC satisfying $d(I(P[A|B]),\mathbb{Z})<\epsilon$ for all matrices $A$ and $B$.\\

We will prove the lemma by induction on $m$. If $m=1$, we set $\vec{\alpha}=[1]$, so $I(P[\vec{\alpha}|\Phi])=I(P)>1-\epsilon$. For $m>1$, from \emph{remark 1} we have:

$$1-\epsilon<I(P)=I\big(P\big[[\vec{e}_1\;...\;\vec{e}_m]\big|\Phi\big]\big)=\sum_{k=1}^m I\big(P\big[\vec{e}_k\big|[\vec{e}_{k+1}\;...\;\vec{e}_m]\big]\big)$$

If $I\big(P\big[\vec{e}_k\big|[\vec{e}_{k+1}\;...\;\vec{e}_m]\big]\big)<\epsilon$ for all $k$, we get $1-\epsilon<m\epsilon \Rightarrow \epsilon>\frac{1}{m+1}$ which is a contradiction, so there exists at least one $k$ satisfying $I\big(P\big[\vec{e}_k\big|[\vec{e}_{k+1}\;...\;\vec{e}_m]\big]\big)>1-\epsilon$. If $k>1$, then $P\big[[\vec{e}_k\;...\;\vec{e}_m]\big|\Phi\big]$ has $m-k+1<m$ users, and $I\big(P\big[[\vec{e}_k\;...\;\vec{e}_m]\big|\Phi\big]\big)\geq I\big(P\big[\vec{e}_k\big|[\vec{e}_{k+1}\;...\;\vec{e}_m]\big]\big)>1-\epsilon$. By induction we get a vector $\vec{\alpha}'\in\mathbb{F}_q^{m-k+1}$ such that $I\big(P\big[[\vec{e}_k\;...\;\vec{e}_m]\big|\Phi\big][\vec{\alpha}'|\Phi]\big)>1-\epsilon$. Let $\vec{\alpha}=[\vec{e}_k\;...\;\vec{e}_m]\vec{\alpha}'$, then by \emph{lemma 2} we have $I(P[\vec{\alpha}|\Phi])=
I\big(P\big[[\vec{e}_k\;...\;\vec{e}_m]\big|\Phi\big][\vec{\alpha}'|\Phi]\big)>1-\epsilon$ and we are done.\\

If $k=1$ then $I\big(P\big[\vec{e_1}\big|[\vec{e}_2\;...\;\vec{e}_m]\big]\big)>1-\epsilon$, so we have (see \emph{remark 1}):
\begin{align*}
I\big(P\big[[\vec{e}_1\;...\;\vec{e}_{m-1}]\big|\vec{e}_m\big]\big)&=I\big(P\big[[\vec{e}_2\;...\;\vec{e}_{m-1}]\big|\vec{e}_m\big]\big)+I\big(P\big[\vec{e_1}\big|[\vec{e}_2\;...\;\vec{e}_m]\big]\big)\\
&\geq I\big(P\big[\vec{e_1}\big|[\vec{e}_2\;...\;\vec{e}_m]\big]\big)>1-\epsilon
\end{align*}
$P\big[[\vec{e}_1\;...\;\vec{e}_{m-1}]\big|\vec{e_m}\big]$ has $m-1$ users. Therefore, by induction we can get a vector $\vec{\alpha}'\in\mathbb{F}_q^{m-1}$ such that $I(P\big[[\vec{e}_1\;...\;\vec{e}_{m-1}]\big|\vec{e}_m\big][\vec{\alpha}'|\Phi]\big)>1-\epsilon$. Let $\vec{\alpha}''=[\vec{e}_1\;...\;\vec{e}_{m-1}]\vec{\alpha}'$, then we have $I(P[\vec{\alpha}''|\vec{e}_m])=I(P\big[[\vec{e}_1\;...\;\vec{e}_{m-1}]\big|\vec{e}_m\big][\vec{\alpha}'|\Phi]\big)>1-\epsilon$ (see \emph{lemma 2}). If $\vec{U}$ and $Y$ are the input and output to the channel $P$ respectively then $I(X;YW)=I(P[\vec{\alpha}''|\vec{e}_m])>1-\epsilon$, where $X=\vec{\alpha}''^T\vec{U}$ and $W=\vec{e}_m^T\vec{U}$ (see \emph{remark 1}). If $I(X;Y)>1-\epsilon$ or $I(W;Y)>1-\epsilon$, we set $\vec{\alpha}=\vec{\alpha}''$ or $\vec{\alpha}=\vec{e}_m$ respectively and we are done.\\

If $I(X;Y)<\epsilon$ and $I(W;Y)<\epsilon$, we have $I(X;YW)>1-\epsilon$, so $H(X|YW)<\epsilon$. $I(XW;Y)=I(X;Y)+I(W;YX)=I(W;X)+I(X;YW)$ which implies $I(W;YX)\geq I(X;YW) - I(X;Y) >1-2\epsilon>\frac{1}{3}>\epsilon$, thus $I(W;YX)>1-\epsilon$ and $H(W|YX)<\epsilon$. Moreover, from the hypothesis we have $I(\beta X+\gamma W;Y)=I(P[\beta\vec{\alpha}''+\gamma\vec{e_m}|\Phi])\notin(\epsilon,1-\epsilon)$ so $H(\beta X+\gamma W|Y)\notin(\epsilon,1-\epsilon)$ for all $\beta,\gamma\in \mathbb{F}_q$.\\

Notice that $\epsilon<\epsilon_m\leq\delta$. Therefore, by the above lemma there exist $\beta',\gamma'\in\mathbb{F}_q$ such that $I(\beta' X+ \gamma' W;Y)>1-\epsilon'=\frac{2}{3}>\epsilon$ which implies $I(P[\vec{\alpha}|\Phi])=I(\beta' X+ \gamma' W;Y)>1-\epsilon$ for $\vec{\alpha}=\beta' \vec{\alpha}''+ \gamma' \vec{e}_m$.\\
\end{proof}

\begin{myprop}
For any $\epsilon<\epsilon_m$, if $P:U_1...U_m\rightarrow Y$ is an $m$-user MAC that satisfies $d(I(P[A|B]),\mathbb{Z})<\epsilon$ for all matrices $A$ and $B$ ($B$ can be $\Phi$), then there exists a matrix $A_P$ of rank $r$ such that $|I({A_P}^T\vec{U};Y)-I(P)|<2\epsilon$, $|I({A_P}^T\vec{U};Y)-r|<\epsilon$ and $|I(P)-r|<\epsilon$.
\end{myprop}
\begin{proof}
Let $V=\{\vec{\alpha}\in\mathbb{F}_q^m:\vec{\alpha}=\vec{0}\;or\;I(\vec{\alpha}^T\vec{U};Y)>1-\epsilon\}$, then $V$ is a subspace of $\mathbb{F}_q^m$:\\

Let $\vec{\alpha}\in V$ and $\beta\in\mathbb{F}_q$, if $\vec{\alpha}=\vec{0}$ or $\beta=0$ then $\beta\vec{\alpha}=\vec{0}\in V$. If $\vec{\alpha}\neq \vec{0}$ and $\beta\neq 0$, then $I(\beta\vec{\alpha}^T\vec{U};Y)=I(\vec{\alpha}^T\vec{U};Y)>1-\epsilon$ and $\beta\vec{\alpha}\in V$.\\

Suppose $m>1$, let $\vec{\alpha}_1,\vec{\alpha}_2\in V$ and $\beta_1,\beta_2\in\mathbb{F}_q$. We can suppose that $\vec{\alpha}_1,\vec{\alpha}_2\neq\vec{0}$ and $\beta_1,\beta_2\neq0$ because otherwise we would be in the previous case. Let $U'=\vec{\alpha}_1^T\vec{U}$, $U''=\vec{\alpha}_2^T\vec{U}$, $X'=\beta_1 U' + \beta_2 U''$ and $X''=\beta_1 U' + (\beta_2+1) U''$ then the transformation $(U',U'')\rightarrow(X',X'')$ is invertible. Thus:
\begin{align*}
I(X';Y)+I(X'';YX')&=I(X'X'';Y)=I(U'U'';Y)=I(U';Y)+I(U'';YU')\\
&\geq I(U';Y)+I(U'';Y) > 2-2\epsilon
\end{align*}

If $I(X';Y)<\epsilon$ then $I(X'';YX')>2-3\epsilon>2-\frac{3}{3}=1$ which is a contradiction (remember that $\epsilon<\epsilon_m\leq\frac{1}{3}$). So $I\big((\beta_1\vec{\alpha}_1+\beta_2\vec{\alpha}_2)^T\vec{U};Y\big)=I(X';Y)>1-\epsilon$, and $\beta_1\vec{\alpha}_1+\beta_2\vec{\alpha}_2\in V$. Therefore, $V$ is a subspace of $\mathbb{F}_q^m$.\\

Let $\vec{\alpha}_1,...,\vec{\alpha}_r$ be a basis of $V$, and let $\vec{\alpha}_{r+1},...,\vec{\alpha}_m$ be $m-r$ vectors extending $\{\vec{\alpha}_1,...,\vec{\alpha}_r\}$ to a basis of $F_q^m$. Define the two matrices $A:=[\vec{\alpha}_1\;...\;\vec{\alpha}_r]$ and $B:=[\vec{\alpha}_{r+1}\;...\;\vec{\alpha}_m]$.\\

If $I(P[B|\Phi])=I(B^T\vec{U};Y)>1-\epsilon$, then by considering the channel $P[B|\Phi]$ we get by \emph{lemma 5} a vector $\vec{\beta}\in\mathbb{F}_q^{m-r}$ such that $I\big(P\big[B\big|\Phi\big][\vec{\beta}|\Phi]\big)>1-\epsilon$, but this means that $I((B\vec{\beta})^T\vec{U};Y)=I(P[B\vec{\beta}|\Phi])=I\big(P\big[B\big|\Phi\big][\vec{\beta}|\Phi]\big)>1-\epsilon$ (see \emph{lemma 2}) and so $B\vec{\beta}\in V$, and therefore $B\vec{\beta}$ can be written as a linear combination of the vectors of $A$ which form a basis for $V$. But this is a contradiction since the vectors of $A$ and $B$ are linearly independent. Therefore, we must have $0\leq I(P[B|\Phi])<\epsilon$.\\

On the other hand, we have:
\begin{align*}
I(A^T\vec{U};Y)=I([\vec{\alpha}_1\;...\;\vec{\alpha}_r]^T\vec{U};Y)=\sum_{k=1}^r I(\vec{\alpha}_k^T\vec{U};Y,[\vec{\alpha}_1\;...\;\vec{\alpha}_{k-1}]^T\vec{U})
\geq \sum_{k=1}^r I(\vec{\alpha}_k^T\vec{U};Y) > r-r\epsilon
\end{align*}

So $r-r\epsilon<I(A^T\vec{U};Y)\leq r$, and $|I(A^T\vec{U};Y)-r|<r\epsilon\leq \frac{m}{m+1}=1-\frac{1}{m+1}<1-\epsilon$ (remember that $\epsilon<\epsilon_m\leq\frac{1}{m+1}$), but $I(A^T\vec{U};Y)=I(P[A|\Phi])$, and from the hypothesis we have $d(I(P[A|\Phi]),\mathbb{Z})<\epsilon$, so $r$ is the closest integer to $I(P[A|\Phi])$, thus $r-\epsilon<I(A^T\vec{U};Y)\leq r$ and $|I(A^T\vec{U};Y)-r|<\epsilon$. Moreover $I(A^T\vec{U};Y,B^T\vec{U})\leq r$ so we get: $$r-\epsilon< I(A^T\vec{U};Y)\leq I(A^T\vec{U};Y,B^T\vec{U}) \leq I(A^T\vec{U};Y,B^T\vec{U}) + I(B^T\vec{U};Y)<r+\epsilon$$ and since the matrix $[A\;B]^T$ is invertible, then we have:
$$I(P)=I(\vec{U};Y)=I([A\; B]^T\vec{U};Y)=I(B^T\vec{U};Y)+I(A^T\vec{U};Y,B^T\vec{U})$$
Therefore, $r-\epsilon< I(P) <r+\epsilon$ since $r-\epsilon< I(A^T\vec{U};Y,B^T\vec{U}) + I(B^T\vec{U};Y) <r+\epsilon$. We conclude that $|I(P)-r|<\epsilon$, and $|I(A^T\vec{U};Y)-I(P)|<2\epsilon$ because we already have $|I(A^T\vec{U};Y)-r|<\epsilon$.
\end{proof}

Now we are ready to prove \emph{theorem 2}:

\begin{proof} \emph{(of theorem 2)}\\

The processes $I(P_n[A|B])$ converge almost surely to integers, and therefore the maximal distance between $I(P_n[A|B])$ and the set of integers converges almost surely to zero. For large enough $n$, this maximal distance becomes less than $\epsilon_m$ and by the previous proposition we conclude that the channels $P_n$ almost surely become almost deterministic linear channels.\\

For the sake of accuracy, we provide the following rigorous proof. Let $P$ be an $m$-user MAC, and $\epsilon>0$. Let $\delta=\frac{1}{2}\min\{\epsilon,\epsilon_m\}$. For $n,l\in \mathbb{N}^\star$, define the event $\mathcal{T}_{n,k}$ as:
\begin{align*}
\mathcal{T}_{n,k}=\Big\{\forall n_1\in[1,m]\cap\mathbb{N},\;&\forall n_2\in[0,m-n_1]\cap\mathbb{N},\; \forall[A\;B]\in\mathbb{F}_q^{m\times(n_1+n_2)}:\\ &\textrm{rank}([A\;B])=n_1+n_2\Rightarrow d(I(P_n[A|B]),\mathbb{Z})<\frac{1}{k}\;\Big\}
\end{align*}
Since the processes $I(P_n[A|B])$ converge almost surely to integer values for all matrices $A$ and $B$ such that $[A\;B]$ is full rank, then the event $\displaystyle\mathcal{T}=\bigcap_{k\geq1}\bigcup_{l\geq1}\bigcap_{n\geq l}\mathcal{T}_{n,k}$ has probability 1. Let $k>\frac{1}{\delta}$, $\displaystyle \textrm{Pr}\Big(\bigcup_{l\geq1}\bigcap_{n\geq l}\mathcal{T}_{n,k}\Big)\geq \textrm{Pr}\Big(\bigcap_{k\geq1}\bigcup_{l\geq1}\bigcap_{n\geq l}\mathcal{T}_{n,k}\Big)=1$. The events $\displaystyle\bigcap_{n\geq l}\mathcal{T}_{n,k}$ are increasing with $l$, so $\displaystyle \textrm{Pr}\Big(\bigcup_{l\geq1}\bigcap_{n\geq l}\mathcal{T}_{n,k}\Big)=\lim_{l\to\infty}\textrm{Pr}\Big(\bigcap_{n\geq l}\mathcal{T}_{n,k}\Big)=1$. We conclude that:
$$\lim_{l\to\infty}\textrm{Pr}(\mathcal{T}_{l,k})\geq\lim_{l\to\infty}\textrm{Pr}\Big(\bigcap_{n\geq l}\mathcal{T}_{n,k}\Big)=1$$
The event $\mathcal{T}_{l,k}$ implies $d(I(P_l[A|B]),\mathbb{Z})<\frac{1}{k}<\delta<\epsilon_m$ for all matrices $A$ and $B$ satisfying $[A\;B]$ is full rank. Then by \emph{proposition 3}, there exists a matrix $A_{P_l}$ of rank $r$ such that $|I({A_{P_l}}^T\vec{U}_{P_l};Y_{P_l})-I({P_l})|<2\delta<\epsilon$, $|I({A_{P_l}}^T\vec{U}_{P_l};Y_{P_l})-r|<\delta<\epsilon$ and $|I({P_l})-r|<\delta<\epsilon$, where $\vec{U}_{P_l}$ and $Y_{P_l}$ are the input and output of $P_l$ respectively. We conclude that the event $\mathcal{T}_{l,k}$ implies the event $\mathcal{C}_l$ defined by:
$$\mathcal{C}_l = \Big\{\exists A_{P_l}\in\mathbb{F}_q^{m\times m},\;\textrm{rank}(A_{P_l})=r_{P_l}, |I(A_{P_l}^T\vec{U}_{P_l};Y_{P_l})-I(P_l)|<\epsilon, |r_{P_l}-I(P_l)|<\epsilon \Big\}$$
So $\displaystyle\lim_{l\to\infty}\textrm{Pr}(\mathcal{C}_l)=1$. By examining the explicit expression of $\textrm{Pr}(\mathcal{C}_l)$ we get:

\begin{align*}
\lim_{l\to\infty} \frac{1}{2^l}\Big|\big\{ & s\in\{-,+\}^l: \exists A_{s}\in\mathbb{F}_q^{m\times r_s}, \textrm{rank}(A_{s})=r_s,\\
& |I(A_{s}^T\vec{U}_{s};Y_{s})-I(P^s)|<\epsilon, |r_s-I(P^s)|<\epsilon \big\}\Big| = 1
\end{align*}
\end{proof}

\section{Rate of polarization}

Now we are interested in the rate of polarization of $P_n$ into deterministic linear channels.

\begin{mydef}
The \emph{Battacharyya parameter} of a single user channel $Q$ with input alphabet $\mathcal{X}$ and output alphabet $\mathcal{Y}$ is defined as:
$$Z(Q)=\frac{1}{|\mathcal{X}|(|\mathcal{X}|-1)}\sum_{\substack{(x,x')\in\mathcal{X}\times\mathcal{X}\\x\neq x'}}\sum_{y\in\mathcal{Y}}\sqrt{Q(y|x)Q(y|x')}$$
\end{mydef}

It's known that $P_e(Q)\leq qZ(Q)$ (see \cite{SasogluTelAri}), where $P_e(Q)$ is the probability of error of the maximum likelihood decoder of $Q$.

\begin{mylem}
Let $P$ be an $m$-user MAC. For any $\vec{\alpha}\in\mathbb{F}_q^m$ we have:
$$\lim_{l\to\infty}\frac{1}{2^l}\Big|\big\{s\in\{-,+\}^l:\;I(P^s[\vec{\alpha}|\Phi])>1-\epsilon, Z(P^s[\vec{\alpha}|\Phi])\geq2^{-{2^{\beta l}}}\big\}\Big|=0$$
for all $0<\epsilon<1$, $0<\beta<\frac{1}{2}$.
\end{mylem}
\begin{proof}
$I(P_n[\vec{\alpha}|\Phi])$ converges almost surely to $0$ or $1$, and this means that $Z(P_n[\vec{\alpha}|\Phi])$ converges also almost surely to 0 or 1 due to the relations between the quantities $I(Q)$ and $Z(Q)$ (see \emph{proposition 3} of \cite{SasogluTelAri}).\\

$Z(P^+[\vec{\alpha}|\Phi])\leq Z(P[\vec{\alpha}|\Phi]^+)$ since $P[\vec{\alpha}|\Phi]^+$ is degraded with respect to $P^+[\vec{\alpha}|\Phi]$, and $Z(P^-[\vec{\alpha}|\Phi])=Z(P[\vec{\alpha}|\Phi]^-)$ since $P[\vec{\alpha}|\Phi]^-$ and $P^-[\vec{\alpha}|\Phi]$ are equivalent. From \cite{SasogluTelAri} we have:
$$Z(P[\vec{\alpha}|\Phi]^-)\leq qZ(P[\vec{\alpha}|\Phi])\;\;\;\;\textrm{and}\;\;\;\;Z(P[\vec{\alpha}|\Phi]^+)=Z(P[\vec{\alpha}|\Phi])^2$$
Now we can apply \emph{theorem 1} of \cite{ArikanTelatar} to get $$\displaystyle\lim_{n\to\infty}\textrm{Pr}\big(I(P_n[\vec{\alpha}|\Phi])>1-\epsilon, Z(P_n[\vec{\alpha}|\Phi])\geq2^{-{2^{n\beta}}}\big)=0$$ by examining the explicit expression of the last probability we get the result.
\end{proof}

\begin{mythe}
The convergence of $P_n$ into deterministic linear channels is almost surely fast:
\begin{align*}
\lim_{l\to\infty} \frac{1}{2^l}\Big|\big\{ & s\in\{-,+\}^l: \exists A_{s}=[\vec{\alpha}_1\;...\;\vec{\alpha}_{r_s}]\in\mathbb{F}_q^{m\times r_s}, \emph{\textrm{rank}}(A_{s})=r_s:\\
& |I(A_{s}^T\vec{U}_{s};Y_{s})-I(P^s)|<\epsilon, |r_s-I(P^s)|<\epsilon, \sum_{k=1}^{r_s} Z(P^s[\vec{\alpha}_k|\Phi])<2^{-2^{\beta l}} \big\}\Big| = 1
\end{align*}
for all $0<\epsilon<1$, $0<\beta<\frac{1}{2}$. $\vec{U}_s$ and $Y_s$ are the input and output of $P^s$ respectively.
\end{mythe}
\begin{proof}
Let $\beta<\beta'<\frac{1}{2}$, define:
$$E_{\vec{\alpha}}=\big\{s\in\{-,+\}^l:\;I(P^s[\vec{\alpha}|\Phi])>1-2\epsilon, Z(P^s[\vec{\alpha}|\Phi])\geq2^{-{2^{\beta' l}}}\big\}\;\;\textrm{for} \;\vec{\alpha}\in\mathbb{F}_q^m$$

\begin{align*}
E_1=\big\{s\in\{-,+\}^l: \exists A_{s}\in\mathbb{F}_q^{m\times r_s}, \textrm{rank}(A_{s})=r_s,|I(A_{s}^T\vec{U}_{s};Y_{s})-I(P^s)|<\epsilon, |r_s-I(P^s)|<\epsilon \big\}
\end{align*}

\begin{align*}
E_2=\big\{ & s\in\{-,+\}^l: \exists A_{s}=[\vec{\alpha}_1\;...\;\vec{\alpha}_{r_s}]\in\mathbb{F}_q^{m\times r_s}, \textrm{rank}(A_{s})=r_s,\\
& |I(A_{s}^T\vec{U}_{s};Y_{s})-I(P^s)|<\epsilon, |r_s-I(P^s)|<\epsilon, \sum_{k=1}^{r_s} Z(P^s[\vec{\alpha}_k|\Phi])<r_s2^{-2^{\beta' l}} \big\}
\end{align*}\\

If $s\in E_1/\big(\bigcup_{\vec{\alpha}\in\mathbb{F}_q^m}E_{\vec{\alpha}}\big)$ then $\exists A_{s}=[\vec{\alpha}_1\;...\;\vec{\alpha}_{r_s}]\in\mathbb{F}_q^{m\times r_s}$ such that $\textrm{rank}(A_{s})=r_s$, $|r_s-I(P^s)|<\epsilon$ and $|I(A_{s}^T\vec{U}_{s};Y_{s})-I(P^s)|<\epsilon$ (so $|I(A_{s}^T\vec{U}_{s};Y_{s})-r_s|<2\epsilon$). For $1\leq k \leq m$ we have:

\begin{align*}
I(A_{s}^T\vec{U}_{s};Y_{s})&=I(P^s[A_s|\Phi])=I(P^s[\vec{\alpha}_k|\Phi])+I\big(P^s\big[[\vec{\alpha}_1\;...\;\vec{\alpha}_{k-1} \vec{\alpha}_{k+1}\;...\;\vec{\alpha}_m]\big|\vec{\alpha}_k\big]\big)\\&>r-2\epsilon\;\forall k\in\{1,...,r_s\}
\end{align*}

and since $I\big(P^s\big[[\vec{\alpha}_1\;...\;\vec{\alpha}_{k-1} \vec{\alpha}_{k+1}\;...\;\vec{\alpha}_m]\big|\vec{\alpha}_k\big]\big)\leq r-1$, then $I(P^s[\vec{\alpha}_k|\Phi])>1-2\epsilon$ which implies $Z(P^s[\vec{\alpha}_k|\Phi])<2^{-2^{\beta' l}}$ since $s\notin E_{\vec{\alpha}_k}$ for all $k\in\{1,...,r_s\}$. So $\displaystyle \sum_{k=1}^{r_s} Z(P^s[\vec{\alpha}_k|\Phi])<r_s2^{-2^{\beta' l}}$, and therefore $s\in E_2$. Thus $\displaystyle E_1/\big(\bigcup_{\vec{\alpha}\in\mathbb{F}_q^m}E_{\vec{\alpha}}\big)\subset E_2$ and $\displaystyle|E_2|\geq|E_1|-\sum_{\vec{\alpha}\in\mathbb{F}_q^m}|E_{\vec{\alpha}}|$. By \emph{theorem 2} and \emph{lemma 6} we have:

$$1\geq\lim_{l\to\infty}\frac{1}{2^l}|E_2|\geq\lim_{l\to\infty}\frac{1}{2^l}\big(|E_1|-\sum_{\vec{\alpha}\in\mathbb{F}_q^m}|E_{\vec{\alpha}}|\big)=1-0=1$$

By noticing that $r_s2^{-2^{\beta' l}}\leq m2^{-2^{\beta' l}}<2^{-2^{\beta l}}$ for $l$ large enough, we conclude the limit in the theorem.
\end{proof}

\section{Polar codes construction}

Choose $0<\epsilon<1$ and $0<\beta<\beta'<\frac{1}{2}$, let $l$ be an integer such that $q2^l2^{-2^{\beta' l}}<2^{-2^{\beta l}}$ and $\frac{1}{2^l}|E_l|>1-\frac{\epsilon}{2m}$, where

\begin{align*}
E_l = \big\{ & s\in\{-,+\}^l: \exists r_s,\exists A_{s}=[\vec{\alpha}_1\;...\;\vec{\alpha}_{r_s}]\in\mathbb{F}_q^{m\times r_s}, \textrm{rank}(A_{s})=r_s:\\
& |I(A_{s}^T\vec{U}_{s};Y_{s})-I(P^s)|<\frac{\epsilon}{2}, |r_s-I(P^s)|<\frac{\epsilon}{2}, \sum_{k=1}^{r_s} Z(P^s[\vec{\alpha}_k|\Phi])<2^{-2^{\beta' l}} \big\}
\end{align*}

Such an integer exists due to \emph{theorem 3}.\\

For each $s\in\{-,+\}^l$, if $s\notin E_l$ set $F(s,k)=1\;\forall k\in\{1,...,k\}$, and if $s\in E_l$ choose a matrix $A_s$ of rank $r_s$ which satisfies the conditions in $E_l$, then choose a set of $r_s$ indices $S_s=\{i_1,...i_{r_s}\}$ such that the corresponding rows of $A_s$ are linearly independent then set $F(s,k)=1$ if $k\notin S_s$, and $F(s,k)=0$ if $k\in S_s$. $F(s,k)=1$ indicates that the user $k$ is frozen in the channel $P^s$, i.e. no useful information is being sent.\\

A polar code is constructed as follows: The user $k$ sends a symbol $U_{s,k}$ through a channel equivalent to $P^s$. If $F(s,k)=0$, $U_{s,k}$ is an information symbol, and if $F(s,k)=1$, $U_{s,k}$ is a certain frozen symbol. Since we are free to choose any value for the frozen symbols, we will analyze the performance of the polar code averaged on all the possible choices of the frozen symbols, so we will consider that $U_{s,k}$ are independent random variables, uniformly distributed in $\mathbb{F}_q$ $\forall s\in\{-,+\}^l,\forall k\in\{1,...,m\}$. However, the value of $U_{s,k}$ will be revealed to the receiver if $F(s,k)=1$, and if $F(s,k)=0$ the receiver has to estimate $U_{s,k}$ from the output of the channel.\\

We associate the set $\{-,+\}^l$ with the strict total order $<$ defined as $s_1...s_l<s_1'...s_l'$ if and only if $s_i=-,s_i'=+$ for some $i\in\{1,...,l\}$ and $s_h=s_h'\;\forall h>i$.

\subsection{Encoding}
Let $\{P_s\}_{s\in\{-,+\}^l}$ be a set of $2^l$ independent copies of the channel $P$. Do not confuse $P_s$ with $P^s$, $P_s$ is a copy of the channel $P$ and $P^s$ is the polarized channel obtained from $P$ as before.\\

Define $U_{s_1,s_2,k}$ for $s_1\in\{-,+\}^{l'},s_2\in\{-,+\}^{l-l'}$, $0\leq l'\leq l$ inductively as:
\begin{itemize}
  \item $U_{\Phi,s,k}=U_{s,k}$ if $l'=0$, $s\in\{-,+\}^l$.
  \item $U_{(s_1;-),s_2,k}=U_{s_1,(s_2;+),k}+U_{s_1,(s_2;-),k}$ if $l'>0$, $s_1\in\{-,+\}^{l'-1}$, $s_2\in\{-,+\}^{l-l'}$.
  \item $U_{(s_1;+),s_2,k}=U_{s_1,(s_2;+),k}$ if $l'>0$, $s_1\in\{-,+\}^{l'-1}$, $s_2\in\{-,+\}^{l-l'}$.
\end{itemize}

The user $k$ sends $U_{s,\Phi,k}$ through the channel $P_s$ for all $s\in\{-,+\}^l$. Let $Y_s$ be the output of the channel $P_s$, and let $Y=\{Y_s\}_{s\in\{-,+\}^l}$. We can prove by induction on $l'$ that the channel $\vec{U}_{s_1,s_2,k}\rightarrow \big(\{Y_s\}_{s\;has\;s_1\; as\;prefix},\{\vec{U}_{s'}\}_{s'<s_2}\big)$ is equivalent to $P^{s_2}$. In particular, the channel $\vec{U}_{s}\rightarrow \big(Y,\{\vec{U}_{s'}\}_{s'<s}\big)$ is equivalent to the channel $P^s$.

\subsection{Decoding}
If $s\notin E_l$ then $F(s,k)=1$ for all $k$, and the receiver knows all $U_{s,k}$, there is nothing to decode. Suppose that $s\in E_l$, if we know $\{\vec{U}_{s'}\}_{s'<s}$ then we can estimate $\vec{U}_{s}$ as follows:

\begin{itemize}
  \item If $F(s,k)=1$ then we know $U_{s,k}$.
  \item We have $F(s,k)=0$ for $r_s$ values of $k$ corresponding to $r_s$ linearly independent rows of $A_s$. So if we know $A_s^T\vec{U}_s$, we can recover $U_{s,k}$ for $F(s,k)=0$.
  \item If $A_s=[\vec{\alpha}_1\;...\;\vec{\alpha}_{r_s}]$, then we can estimate $A_s^T\vec{U}_s$ by estimating $\vec{\alpha}_h^T\vec{U}_s$ for $h\in\{1,...,r_s\}$.
  \item Since $\vec{\alpha}_h^T\vec{U}_s\rightarrow \big(Y,\{U_{s',k}\}_{s'<s}\big)$ is equivalent to $P^s[\vec{\alpha}_h|\Phi]$, we can estimate $\vec{\alpha}_h^T\vec{U}_s$ using the maximum likelihood decoder of $P^s[\vec{\alpha}_h|\Phi]$.
  \item Let $\mathcal{D}_s(Y,\{\vec{U}_{s'}\}_{s'<s})$ be the estimate of $\vec{U}_s$ obtained from $(Y,\{\vec{U}_{s'}\}_{s'<s})$ by the above procedure.
\end{itemize}

This motivates the following successive cancelation decoder:

\begin{itemize}
  \item $\hat{\vec{U}}_s=\vec{U}_s$ if $s\notin E_l$.
  \item $\hat{\vec{U}}_s=\mathcal{D}_s(Y,\{\hat{\vec{U}}_{s'}\}_{s'<s})$ if $s\in E_l$.
\end{itemize}

\subsection{Performance of polar codes}

If $s\in E_l$, the probability of error in estimating $\vec{\alpha}_h^T\vec{U}_s$ using the maximum likelihood decoder is upper bounded by $qZ(P^s[\vec{\alpha}_h|\Phi])$. So the probability of error in estimating $A_s^T\vec{U}_s$ is upper bounded by $\displaystyle \sum_{k=1}^{r_s} qZ(P^s[\vec{\alpha}_k|\Phi])<q2^{-2^{\beta' l}}$. Therefore, the probability of error in estimating $\vec{U}_s$ from $(Y,\{\vec{U}_{s'}\}_{s'<s})$ is upper bounded by $q2^{-2^{\beta' l}}$ when $s\in E_l$\\

Note that $\mathcal{D}_s(Y,\{\vec{U}_{s'}\}_{s'<s})=\vec{U}_s,\;(\forall s\in E_l) \Leftrightarrow \mathcal{D}_s(Y,\{\hat{\vec{U}}_{s'}\}_{s'<s})=\vec{U}_s\;(\forall s\in E_1)$, so the probability of error of the above successive cancelation decoder is upper bounded by $$\sum_{s\in E_l} \textrm{Pr}\big(\mathcal{D}_s(Y,\{\vec{U}_{s'}\}_{s'<s})\neq \vec{U}_{s}\big) < |E_l|q2^{-2^{\beta'l}}\leq q2^l2^{-2^{-\beta'l}} < 2^{-2^{\beta l}}$$ The above upper bound was calculated on average over a random choice of the frozen symbols. Therefore, there is at least one choice of the frozen symbols for which the upper bound of the probability of error still holds.\\

The last thing to discuss is the rate vector of polar codes. The rate at which the user $k$ is communicating is $\displaystyle R_k=\frac{1}{2^l}\sum_{s\in E_l}\big(1-F(s,k)\big)$, the sum rate is:
$$\displaystyle R=\sum_{1\leq k\leq m}R_k=\frac{1}{2^l}\sum_{1\leq k\leq m}\sum_{s\in E_l}\big(1-F(s,k)\big)=\frac{1}{2^l}\sum_{s\in E_l}r_s$$
We have $|I(P^s)-r_s|<\frac{\epsilon}{2}$ and $I(P^s)<r_s+\frac{\epsilon}{2}$ for all $s\in E_l$. And since we have $\displaystyle\sum_{s\in\{-,+\}^l} I(P^s)=2^lI(P)$ we conclude:
\begin{align*}
I(P)&=\frac{1}{2^l}\sum_{s\in \{-,+\}^l} I(P^s)= \frac{1}{2^l}\sum_{s\in E_l}I(P^s) + \frac{1}{2^l}\sum_{s\in E_l^c}I(P^s)<\frac{1}{2^l}\sum_{s\in E_l}(r_s+\frac{\epsilon}{2}) + \frac{1}{2^l}|E_l^c|m \\&< R + \frac{1}{2^l}|E_l|\frac{\epsilon}{2} + m\frac{\epsilon}{2m}\leq R+\frac{\epsilon}{2}+\frac{\epsilon}{2}=R+\epsilon
\end{align*}

To this end we have proved the following theorem which is the main result in this report:

\begin{mythe}
For every $0<\epsilon<1$ and for every $0<\beta<\frac{1}{2}$, there exists a polar code of length $N$ having a sum rate $R>I(P)-\epsilon$ and a probability of error $P_e<2^{-N^\beta}$.
\end{mythe}

A final note to report is that by changing our choice of the indices in $S_s$, the rate vector of the polar code moves at a distance of at most $\epsilon$ along the dominant face of the capacity region achievable by polar codes. However, the dominant face of the initial capacity region can be strictly bigger than the dominant face achievable by polar codes.

\section{Case study}

In this section, we are interested in studying the problem of loss in the capacity region by polarization in the special case of channels which are combination of deterministic linear channels.

\begin{mydef}
An $m$-user MAC $P$ is said to be a combination of $n$ linear channels, if there are $n$ matrices $A_1,...,A_n$, $(A_k\in\mathbb{F}_q^{m\times m_k})$ such that $P$ is equivalent to the channel $\displaystyle P_l:\mathbb{F}_q\times...\times\mathbb{F}_q\rightarrow\bigcup_{k=1}^n \{k\}\times \mathbb{F}_q^{m_k}$ defined by:

\begin{equation*}
P_l(k,\vec{y}|\vec{x})=
\begin{cases}
p_k\;&\emph{\textrm{if}}\;A_k^T\vec{x}=\vec{y}\\
0\;&\emph{\textrm{otherwise}}
\end{cases}
\;\;\forall k\in\{1,...,n\},
\forall \vec{y}\in\mathbb{F}_q^{m_k},\forall\vec{x}\in\mathbb{F}_q^m
\end{equation*}

where $\displaystyle \sum_{k=1}^n p_k=1$ and $p_k\neq0\;\forall k$. The channel $P_l$ is denoted by $P_l=\displaystyle\sum_{k=1}^n p_k \mathcal{C}_{A_k}$.
\end{mydef}

The channel $P_l$ can be seen as a box where we have a collection of matrices. At each channel use, a matrix $A_k$ from the box is chosen randomly according to the probabilities $p_k$. The output of the channel is $A_k^T\vec{x}$, together with the index $k$ (so the receiver knows which matrix has been used).

\subsection{Characterizing non-losing channels}

In the case of channels that are combination of linear channels, we are interested in finding in the channels whose capacity region is preserved upon the polarization process.

\begin{myprop}
If $\{A_k,A_k':1\leq k\leq n\}$ is a set of matrices such that $\textrm{\emph{span}}(A_k)=\textrm{\emph{span}}(A_k')\;\forall k\in\{1,...m\}$, then the two channels $\displaystyle P=\sum_{k=1}^n p_k \mathcal{C}_{A_k}$ and $\displaystyle P'=\sum_{k=1}^n p_k \mathcal{C}_{A_k'}$ are equivalent.
\end{myprop}
\begin{proof}
If $\textrm{span}(A_k)=\textrm{span}(A_k')$, we can determine $A_k^T \vec{x}$ from ${A'_k}^T \vec{x}$ and vice versa. Therefore, from the output of $P$, we can deterministically obtain the output of $P'$ and vice versa. In this sense, $P$ and $P'$ are equivalent, and have the same capacity region.
\end{proof}

\begin{mynot}
Motivated by the above proposition, we will write $\displaystyle P\equiv\sum_{k=1}^n p_k\mathcal{C}_{V_k}$ where $\{V_k:1\leq k\leq n\}$ is a set of $n$ subspaces of $\mathbb{F}_q^m$, whenever $P$ is equivalent to $\displaystyle \sum_{k=1}^n p_k\mathcal{C}_{A_k}$ and $\textrm{\emph{span}}(A_k)=V_k$.
\end{mynot}

\begin{myprop}
If $\displaystyle P\equiv\sum_{k=1}^np_k\mathcal{C}_{V_k}$, then $\displaystyle I[S](P)=\sum_{k=1}^n p_k\emph{\textrm{dim}}\big(\emph{\textrm{proj}}_S(V_k)\big)$ for all $S\subset\{1,...,m\}$. Where $\emph{\textrm{proj}}_S$ denotes the canonical projection on $\mathbb{F}_q^S$ defined by $\emph{\textrm{proj}}_S(\vec{x})=\emph{\textrm{proj}}_S(x_1,...,x_m)=(x_{i_1},...,x_{i_{|S|}})$ for $\vec{x}=(x_1,...,x_m)\in \mathbb{F}_q^m$ and $S=\{i_1,...,i_{|S|}\}$.
\end{myprop}
\begin{proof}
Let $X_1,...,X_m$ be the input to the channel $\sum_{k=1}^n p_k\mathcal{C}_{A_k}$ (where $A_k$ spans $V_k$), and let $K,Y$ be the output of it. We have:\\

$H\big(X(S)|K,Y,X(S^c)\big)$
\begin{align*}
&=\sum_{k,\vec{y}} \textrm{Pr}(k,\vec{y})H\big(X(S)|k,\vec{y},X(S^c)\big)=\sum_{k,\vec{y}}\sum_{\vec{x}}\textrm{Pr}(k,\vec{y}|\vec{x})
\textrm{Pr}(\vec{x}) H\big(X(S)|A_k^T\vec{X},X(S^c)\big)\\
&=\sum_{k,\vec{y}}\sum_{\substack{\vec{x},\\A_k^T\vec{x}=\vec{y}}}p_k\textrm{Pr}(\vec{x}) H\big(X(S)|k,\vec{y},X(S^c)\big)=\sum_{k} p_k H\big(X(S)|A_k^T\vec{X},X(S^c)\big)\\
&=\sum_{k} p_k H\big(X(S)|A_k(S)^T\vec{X}(S),X(S^c)\big) = \sum_{k} p_k H\big(X(S)|A_k(S)^T\vec{X}(S)\big)
\end{align*}
Where $A_k(S)$ is obtained from $A_k$ by taking the rows corresponding to $S$. For a given value of $A_k(S)^T\vec{X}(S)$, we have $q^{d_k}$ possible values of $\vec{X}(S)$ with equal probabilities, where $d_k$ is the dimension of the null space of the mapping $\vec{X}(S)\rightarrow A_k(S)^T\vec{X}(S)$, so we have $H\big(X(S)|A_k(S)^T\vec{X}(S)\big)=d_k$.\\

On the other hand, $|S|-H\big(X(S)|A_k(S)^T\vec{X}(S)\big)=|S|-d_k$ is the dimension of the range space of the the mapping $\vec{X}(S)\rightarrow A_k(S)^T\vec{X}(S)$, which is also equal to the rank of $A_k(S)^T$. Therefore, we have:
\begin{align*}
|S|-H\big(X(S)|A_k(S)^T\vec{X}(S)\big)&=
\textrm{rank}(A_k(S)^T)=\textrm{rank}\big(A_k(S)\big)=\textrm{dim}\Big(\textrm{span}\big(A_k(S)\big)\Big)\\
&=\textrm{dim}\Big(\textrm{span}\big
(\textrm{proj}_S(A_k)\big)\Big)=\textrm{dim}\Big(\textrm{proj}_S\big(\textrm{span}(A_k)\big)\Big)\\&=\textrm{dim}\big(\textrm{proj}_S(V_k)\big)
\end{align*}
We conclude:
\begin{align*}
I(X(S);K,Y,X(S^c))&=H(X(S))-H(X(S)|K,Y,X(S^c))\\
&=|S|-\sum_{k} p_k H(X(S)|A_k(S)^T\vec{X}(S))\\
&=\sum_{k} p_k \big(|S|-H(X(S)|A_k(S)^T\vec{X}(S))\big)\\
&=\sum_{k} p_k (|S|-d_k)=\sum_{k} p_k \textrm{dim}(\textrm{proj}_S(V_k))
\end{align*}
\end{proof}

\begin{myprop}
If $\displaystyle P\equiv\sum_{k=1}^np_k\mathcal{C}_{V_k}$ then:
\begin{itemize}
  \item $\displaystyle P^-\equiv\sum_{k_1=1}^n\sum_{k_2=1}^np_{k_1}p_{k_2}\mathcal{C}_{V_{k_1}\cap V_{k_2}}$
  \item $\displaystyle P^+\equiv\sum_{k_1=1}^n\sum_{k_2=1}^np_{k_1}p_{k_2}\mathcal{C}_{V_{k_1} + V_{k_2}}$
\end{itemize}
\end{myprop}
\begin{proof}
Suppose without lost of generality that $\displaystyle P=\sum_{k=1}^n p_k\mathcal{C}_{A_k}$ where $A_k$ spans $V_k$. Let $\vec{U}_1$ be an arbitrarily distributed random vector in $\mathbb{F}_q^m$ (not necessarily uniform), let $\vec{U}_2$ be a uniformly distributed random vector in $\mathbb{F}_q^m$ independent from $U_1$. Let $\vec{X}_1=\vec{U}_1+\vec{U}_2$ and $\vec{X}_2=\vec{U}_2$. Let $(K_1,A_{K_1}^T\vec{X}_1)$ and $(K_2,A_{K_2}^T\vec{X}_2)$ be the output of $P$ when the input is $X_1$ and $X_2$ respectively. Then the channel $\vec{U_1}\rightarrow(K_1,A_{K_1}^T\vec{X}_1,K_2,A_{K_2}^T\vec{X}_2)$ is equivalent to $P^-$ with $\vec{U_1}$ as input. We did not put any constraint on the distribution of $\vec{U}_1$ (such as saying that $\vec{U}_1$ is uniform) because in general, the model of a channel is characterized by its conditional probabilities and no assumption is made on the input probabilities.\\

Fix $K_1=k_1$ and $K_2=k_2$, let $A_{k_1\wedge k_2},B_{k_1}$ and $B_{k_2}$ be three matrices chosen such that $A_{k_1\wedge k_2}$ spans $V_{k_1}\cap V_{k_2}$, $A_{k_1}=[A_{k_1\wedge k_2}\;B_{k_1}]$ spans $V_{k_1}$, $A_{k_2}=[A_{k_1\wedge k_2}\;B_{k_2}]$ spans $V_{k_2}$, and the columns of $[A_{k_1\wedge k_2}\;B_{k_1}\;B_{k_2}]$ are linearly independent. Then knowing $A_{k_1}^T\vec{X}_1$ and $A_{k_2}^T\vec{X}_2$ is equivalent to knowing $A_{k_1\wedge k_2}^T(\vec{U}_1+\vec{U}_2)$, $B_{k_1}^T(\vec{U}_1+\vec{U}_2)$, $A_{k_1\wedge k_2}^T\vec{U}_2$ and $B_{k_2}^T\vec{U}_2$, which is equivalent to knowing $\vec{T}_{k_1,k_2}^1=A_{k_1\wedge k_2}^T\vec{U}_1$, $\vec{T}_{k_1,k_2}^2=B_{k_1}^T(\vec{U}_1+\vec{U}_2)$ and $\vec{T}_{k_1,k_2}^3=[A_{k_1\wedge k_2}\;B_{k_2}]^T\vec{U}_2$. We conclude that $P^-$ is equivalent to the channel:
$$\vec{U}_1\rightarrow \big(K_1,K_2,\vec{T}_{K_1,K_2}^1,\vec{T}_{K_1,K_2}^2,\vec{T}_{K_1,K_2}^3\big)$$
Conditioned on $(K_1,K_2,\vec{T}_{K_1,K_2}^1)$ we have $[B_{K_1}\;A_{K_1\wedge K_2}\;B_{K_2}]^T\vec{U}_2$ is uniform (since the matrix $[B_{K_1}\;A_{K_1\wedge K_2}\;B_{K_2}]$ is full rank) and independent from $\vec{U}_1$, so $[A_{K_1\wedge K_2}\;B_{K_2}]^T\vec{U}_2$ is independent from $(B_{K_1}^T\vec{U}_2,\vec{U}_1)$, which implies that $[A_{K_1\wedge K_2}\;B_{K_2}]^T\vec{U}_2$ is independent from $\big(B_{K_1}^T(\vec{U}_1+\vec{U}_2),\vec{U}_1\big)$. Also conditioned on $(K_1,K_2,\vec{T}_{K_1,K_2}^1)$, $B_{K_1}^T\vec{U}_2$ is uniform and independent from $\vec{U}_1$, which implies that $\vec{U}_1$ is independent from $B_{K_1}^T(\vec{U}_1+\vec{U}_2)$, and this is because the columns of $B_{K_1}$ and $A_{K_1\wedge K_2}$ are linearly independent. We conclude that conditioned on $(K_1,K_2,\vec{T}_{K_1,K_2}^1)$, $\vec{U}_1$ is independent from $\big(\vec{T}_{K_1,K_2}^2,\vec{T}_{K_1,K_2}^3\big)$. Therefore, $\big(K_1,K_2,\vec{T}_{K_1,K_2}^1\big) =\big(K_1,K_2,A_{k_1\wedge k_2}^T\vec{U}_1\big)$ form sufficient statistics. We conclude that $P^-$ is equivalent to the channel:

$$\vec{U}_1\rightarrow =\big(K_1,K_2,A_{k_1\wedge k_2}^T\vec{U}_1\big)$$

And since $\textrm{Pr}(K_1=k_1,K_2=k_2)=p_{k_1}p_{k_2}$, and $A_{k_1\wedge k_2}$ spans $V_{k_1}\cap V_{k_2}$ we conclude that $\displaystyle P^-\equiv\sum_{k_1=1}^n\sum_{k_2=1}^np_{k_1}p_{k_2}\mathcal{C}_{V_{k_1}\cap V_{k_2}}$.\\

Now let $\vec{U}_2$ be arbitrarily distributed in $\mathbb{F}_q^m$ (not necessarily uniform) and $\vec{U}_1$ be a uniformly distributed random vector in $\mathbb{F}_q^m$ independent from $\vec{U}_2$. Let $\vec{X}_1=\vec{U}_1+\vec{U}_2$ and $\vec{X}_2=\vec{U}_2$. Let $(K_1,A_{K_1}^T\vec{X}_1)$ and $(K_2,A_{K_2}^T\vec{X}_2)$ be the output of $P$ when the input is $X_1$ and $X_2$ respectively. Then the channel $\vec{U_2}\rightarrow(K_1,A_{K_1}^T\vec{X}_1,K_2,A_{K_2}^T\vec{X}_2,\vec{U_1})$ is equivalent to $P^+$ with $\vec{U_2}$ as input. Note that the uniform distribution constraint is now on $\vec{U}_1$ and no constraint is put on the distribution of $\vec{U}_2$, since now $\vec{U}_2$ is the input to the channel $P^+$.\\

Knowing $A_{K_1}^T\vec{X}_1$, $A_{K_2}^T\vec{X}_2$ and $\vec{U_1}$ is equivalent to knowing $A_{K_1}^T(\vec{U}_1+\vec{U}_2)$, $A_{K_2}^T\vec{U}_2$ and $\vec{U_1}$, which is equivalent to knowing $A_{K_1}^T\vec{U}_2$, $A_{K_2}^T\vec{U}_2$ and $\vec{U_1}$. So $P^+$ is equivalent to the channel:

$$\vec{U_2}\rightarrow \big(K_1,K_2,[A_{k_1}\;A_{k_2}]^T\vec{U}_2,\vec{U}_1\big)$$

And since $\vec{U}_1$ is independent from $\vec{U}_2$, the above channel (and hence $P+$) is equivalent to the channel:

$$\vec{U_2}\rightarrow \big(K_1,K_2,[A_{k_1}\;A_{k_2}]^T\vec{U}_2\big)$$

We also have $\textrm{Pr}(K_1=k_1,K_2=k_2)=p_{k_1}p_{k_2}$, and $[A_{k_1}\;A_{k_2}]$ spans $V_{k_1}+V_{k_2}$. We conclude that $\displaystyle P^+\equiv\sum_{k_1=1}^n\sum_{k_2=1}^np_{k_1}p_{k_2}\mathcal{C}_{V_{k_1} + V_{k_2}}$.
\end{proof}

\begin{mylem}
Let $\displaystyle P\equiv\sum_{k=1}^np_k\mathcal{C}_{V_k}$ and $S\subset\{1,...,m\}$, then
\begin{align*}
\frac{1}{2}\big(I[S](P^-)+I[S](P^+)\big)&=I[S](P) \Leftrightarrow \\
&\Big(\forall (k_1,k_2);\;\emph{\textrm{proj}}_S(V_{k_1}\cap V_{k_2})=\emph{\textrm{proj}}_S(V_{k_1})\cap\emph{\textrm{proj}}_S(V_{k_2})\Big)
\end{align*}
\end{mylem}
\begin{proof}
We know that if $V$ and $V'$ are two subspaces of $\mathbb{F}_q^m$, then $\textrm{proj}_S(V\cap V')\subset\textrm{proj}_S(V)\cap\textrm{proj}_S(V')$ and $\textrm{proj}_S(V+V')=\textrm{proj}_S(V)+\textrm{proj}_S(V')$, which implies that:
\begin{align*}
\textrm{dim}\big(\textrm{proj}_S(V\cap V')\big)&\leq\textrm{dim}\big(\textrm{proj}_S(V)\cap\textrm{proj}_S(V')\big)\\
\textrm{dim}\big(\textrm{proj}_S(V+ V')\big)&=\textrm{dim}\big(\textrm{proj}_S(V)+\textrm{proj}_S(V')\big)
\end{align*}

We conclude:\\

$\textrm{dim}\big(\textrm{proj}_S(V\cap V')\big) + \textrm{dim}\big(\textrm{proj}_S(V+ V')\big)$
\vspace*{-2mm}
\begin{align*}
&\leq\textrm{dim}\big(\textrm{proj}_S(V)\cap\textrm{proj}_S(V')\big)+\textrm{dim}\big(\textrm{proj}_S(V)+\textrm{proj}_S(V')\big)\\
&=\textrm{dim}\big(\textrm{proj}_S(V)\big) + \textrm{dim}\big(\textrm{proj}_S(V')\big)
\end{align*}

Therefore:\\

$\frac{1}{2}\big(I[S](P^-)+I[S](P^+)\big)$
\begin{align*}
&=\frac{1}{2}\Big(\sum_{k_1=1}^n\sum_{k_2=1}^n p_{k_1}p_{k_2}\textrm{dim}\big(\textrm{proj}_S(V_{k_1}\cap V_{k_2})\big) + \sum_{k_1=1}^n\sum_{k_2=1}^n p_{k_1}p_{k_2}\textrm{dim}\big(\textrm{proj}_S(V_{k_1}+ V_{k_2})\big)\Big)\\
&=\frac{1}{2}\Big(\sum_{k_1=1}^n\sum_{k_2=1}^n p_{k_1}p_{k_2}\big(\textrm{dim}\big(\textrm{proj}_S(V_{k_1}\cap V_{k_2})\big)+\textrm{dim}\big(\textrm{proj}_S(V_{k_1}+ V_{k_2})\big)\big)\Big)\\
&\leq \frac{1}{2}\Big(\sum_{k_1=1}^n\sum_{k_2=1}^n p_{k_1}p_{k_2}\big(\textrm{dim}\big(\textrm{proj}_S( V_{k_1})\big)+\textrm{dim}\big(\textrm{proj}_S(V_{k_2})\big)\big)\Big)\\
&= \frac{1}{2}\Big(\sum_{k_1=1}^n p_{k_1}\textrm{dim}\big(\textrm{proj}_S( V_{k_1})\big)+ \sum_{k_2=1}^n p_{k_2}\textrm{dim}\big(\textrm{proj}_S(V_{k_2})\big)\Big)\\
&=\frac{1}{2}(I[S](P)+I[S](P))=I[S](P)
\end{align*}

So if we have $\textrm{proj}_S(V_{k_1}\cap V_{k_2})\subsetneq\textrm{proj}_S(V_{k_1})\cap\textrm{proj}_S(V_{k_2})$ for some $k_1,k_2$, then we have $\textrm{dim}\big(\textrm{proj}_S(V_{k_1}\cap V_{k_2})\big)<\textrm{dim}\big(\textrm{proj}_S(V_{k_1})\cap\textrm{proj}_S(V_{k_2})\big)$, and the above inequality of mutual information will be strict. We conclude that:

\begin{align*}
\frac{1}{2}\big(I[S](P^-)+I[S](P^+)\big)&=I[S](P) \Leftrightarrow \\
&\Big(\forall (k_1,k_2);\;\textrm{proj}_S(V_{k_1}\cap V_{k_2})=\textrm{proj}_S(V_{k_1})\cap\textrm{proj}_S(V_{k_2})\Big)
\end{align*}
\end{proof}

\begin{mydef}
Let $\mathcal{V}$ be a set of subspaces of $\mathbb{F}_q^m$, we define the closure of $\mathcal{V}$, $cl(\mathcal{V})$, as being the minimal set of subspaces of $\mathbb{F}_q^m$ closed under the two operations $\cap$ and $+$, and including $\mathcal{V}$. We say that the set $\mathcal{V}$ is consistent with respect to $S\subset\{1,...,m\}$ if and only if it satisfies the following property:
$$\Big(\forall (V_1,V_2)\in cl(\mathcal{V});\;\textrm{proj}_S(V_{k_1}\cap V_{k_2})=\textrm{proj}_S(V_{k_1})\cap\textrm{proj}_S(V_{k_2})\Big)$$
\end{mydef}

\begin{mycor}
If $\mathcal{V}=\{V_k:1\leq k\leq n\}$. $I[S](P)$ is preserved upon the polarization process if and only if $\mathcal{V}$ is consistent with respect to $S$.
\end{mycor}
\begin{proof}
Upon the polarization process, we are performing successively the $\cap$ and $+$ operators, which means that we'll reach the closure of $\mathcal{V}$ after a finite number of steps. So $I[S](P)$ is preserved if and only if the above lemma applies to $cl(\mathcal{V})$.
\end{proof}

The above corollary gives a characterization for a combination of linear channels to preserve $I[S](P)$. However, this characterization involves using the closure operator. The next proposition gives a sufficient condition that uses only the initial configuration of subspaces $\mathcal{V}$. This proposition gives some ``geometric'' view of what the subspaces should look like if we don't want to lose.\\

\begin{myprop}
If there exists a subspace $V_S$ of dimension $|S|$ whose projection on $S$ is $\mathbb{F}_q^S$ (i.e. $\emph{proj}_S(V_S)=\mathbb{F}_q^S$), such that for every $V\in\mathcal{V}$ we have $\emph{proj}_S(V_S\cap V)=\emph{proj}_S(V)$, then $I[S](P)$ is preserved upon the polarization process. In other words, if every subspace in $\mathcal{V}$ passes through $V_S$ ``orthogonally'' to $S$, then $I[S](P)$ is preserved upon the polarization process.
\end{myprop}
\begin{proof}
Let $V_S$ be a subspace satisfying the hypothesis, then it satisfies also the hypothesis if we replace $\mathcal{V}$ by it's closure: If $V_1$ and $V_2$ are two arbitrary subspaces satisfying $$\textrm{proj}_S(V_S\cap V_1)=\textrm{proj}_S(V_1)\;\textrm{and}\;\textrm{proj}_S(V_S\cap V_2)=\textrm{proj}_S(V_2)$$ then $\textrm{proj}_S(V_1)\subset \textrm{proj}_S\big(V_S\cap(V_1+V_2)\big)$ and $\textrm{proj}_S(V_2)\subset \textrm{proj}_S\big(V_S\cap(V_1+V_2)\big)$, which implies $\textrm{proj}_S(V_1+V_2)=\textrm{proj}_S(V_1)+\textrm{proj}_S(V_2)\subset \textrm{proj}_S\big(V_S\cap(V_1+V_2)\big)$. Therefore, $\textrm{proj}_S\big(V_S\cap(V_1+V_2)\big)=\textrm{proj}_S(V_1+V_2)$ since the inverse inclusion is trivial.\\

Now let $\vec{x}\in\textrm{proj}_S(V_1) \cap \textrm{proj}_S(V_2)$, then $\vec{x}\in\textrm{proj}_S(V_1)=\textrm{proj}_S(V_1\cap V_S)$ and similarly $\vec{x}\in\textrm{proj}_S(V_2\cap V_S)$ which implies that there are two vectors $\vec{x_1}\in V_1\cap V_S$ and $\vec{x_2}\in V_2\cap V_S$ such that $\vec{x}=\textrm{proj}_S(\vec{x}_1)=\textrm{proj}_S(\vec{x_2})$. And since $\textrm{proj}_S(V_S)=\mathbb{F}_q^S$ and $\textrm{dim}(V_S)=|S|$, then the mapping $\textrm{proj}_S:V_S\rightarrow\mathbb{F}_q^S$ is invertible and so $\vec{x}_1=\vec{x}_2$ which implies that $\vec{x}\in\textrm{proj}_S(V_1\cap V_2\cap V_S)$. Thus $\textrm{proj}_S(V_1) \cap \textrm{proj}_S(V_2)\subset\textrm{proj}_S(V_1\cap V_2)\subset\textrm{proj}_S(V_1\cap V_2\cap V_S)$. We conclude that $\textrm{proj}_S(V_1) \cap \textrm{proj}_S(V_2)=\textrm{proj}_S(V_1\cap V_2)=\textrm{proj}_S(V_1\cap V_2\cap V_S)$ since the inverse inclusions are trivial.\\

We conclude that the set of subspaces $V$ satisfying $\textrm{proj}_S(V\cap V_S)=\textrm{proj}_S(V)$ is closed under the two operators $\cap$ and $+$. And since $\mathcal{V}$ is a subset of this set, $cl(\mathcal{V})$ is a subset as well. Now let $V_1,V_2\in cl(\mathcal{V})$, then $\textrm{proj}_S(V_S\cap V_1)=\textrm{proj}_S(V_1)$ and $\textrm{proj}_S(V_S\cap V_2)=\textrm{proj}_S(V_2)$. Then $\textrm{proj}_S(V_1) \cap \textrm{proj}_S(V_2)=\textrm{proj}_S(V_1\cap V_2)$ as we have seen in the previous paragraph. We conclude that $\mathcal{V}$ is consistent with respect to $S$ and so $I[S](P)$ is preserved.
\end{proof}

\begin{myconj}
The condition in proposition 7 is necessary.
\end{myconj}

\subsection{Total loss in the dominant face}

After characterizing the non-losing channels, we are now interested in studying the amount of loss in the capacity region. In order to simplify the problem, we only study it in the case of binary input 2-user MAC since we can easily generalize for the general case.\\

Since we only have 5 subspaces of $\mathbb{F}_2^2$, we write $\displaystyle P\equiv \sum_{k=0}^4 p_k\mathcal{C}_{V_k}$, where $V_0$, ..., $V_4$ are the 5 possible subspaces of $\mathbb{F}_2^2$:
\begin{align*}
V_0&=\{(0,0)\}\\
V_1&=\{(0,0),(1,0)\}\\
V_2&=\{(0,0),(0,1)\}\\
V_3&=\{(0,0),(1,1)\}\\
V_4&=\{(0,0),(1,0),(0,1),(1,1)\}
\end{align*}
We have $I[{1}](P)=p_1+p_3+p_4$, $I[{2}](P)=p_2+p_3+p_4$ and $I(P)=I[{1,2}](P)=p_1+p_2+p_3+2p_4$.

\begin{mydef}
Let $\displaystyle P\equiv\sum_{k=0}^4 p_k\mathcal{C}_{V_k}$ and $s\in\{-,+\}^l$, we write $p_k^s$ to denote the component of $V_k$ in $P^s$, i.e. we have $\displaystyle P^s\equiv\sum_{k=0}^4 p_k^s\mathcal{C}_{V_k}$.\\

We denote the average of $p_k^s$ on all possible $s\in\{-,+\}^l$ by $p_k^{(l)}$. i.e. $\displaystyle p_k^{(l)}=\frac{1}{2^l}\sum_{s\in\{-,+\}^l}p_k^s$. $p_k^{(\infty)}$ is the limit of $p_k^{(l)}$ as $l$ tends to infinity.\\

We denote the average of $I[{1}](P^s)$ (resp. $I[{2}](P^s)$ and $I(P^s)$) on all possible $s\in\{-,+\}^l$ by $I_1^{(l)}$ (resp. $I_2^{(l)}$ and $I^{(l)}$). We have $I_1^{(l)}=p_1^{(l)}+p_3^{(l)}+p_4^{(l)}$, $I_2^{(l)}=p_2^{(l)}+p_3^{(l)}+p_4^{(l)}$ and $I^{(l)}=p_1^{(l)}+p_2^{(l)}+p_3^{(l)}+2p_4^{(l)}$. If $l$ tends to infinity we get $I_1^{(\infty)}=p_1^{(\infty)}+p_3^{(\infty)}+p_4^{(\infty)}$, $I_2^{(\infty)}=p_2^{(\infty)}+p_3^{(\infty)}+p_4^{(\infty)}$ and $I^{(\infty)}=p_1^{(\infty)}+p_2^{(\infty)}+p_3^{(\infty)}+2p_4^{(\infty)}$.
\end{mydef}

\begin{mydef}
We say that we have total loss in the dominant face in the polarization process, if the dominant face of the capacity region converges to a single point.
\end{mydef}

\begin{myrem}
The symmetric capacity region after $l$ polarization steps is the average of the symmetric capacity regions of all the channels $P^s$ obtained after $l$ polarization steps ($s\in\{-,+\}^l$). Therefore, this capacity region is given by:

$$\mathcal{J}(P^{(l)})=\{(R_1,R_2):\;0\leq R_1\leq I_1^{(l)},\;0\leq R_2\leq I_2^{(l)},\;0\leq R_1+R_2\leq I^{(l)}\}$$

The above capacity region converges to the ``final capacity region'':

$$\mathcal{J}(P^{(\infty)})=\{(R_1,R_2):\;0\leq R_1\leq I_1^{(\infty)},\;0\leq R_2\leq I_2^{(\infty)},\;0\leq R_1+R_2\leq I^{(\infty)}\}$$

The dominant face converges to a single point if and only if $I^{(\infty)}=I_1^{(\infty)}+I_2^{(\infty)}$, which is equivalent to $p_1^{(\infty)}+p_2^{(\infty)}+ p_3^{(\infty)}+2p_4^{(\infty)}=p_1^{(\infty)}+p_2^{(\infty)}+2p_3^{(\infty)}+2p_4^{(\infty)}$. We conclude that we have total loss in the dominant face if and only if $p_3^{(\infty)}=0$.
\end{myrem}

\begin{mylem}
The order of $p_1,p_2$ and $p_3$ remains the same upon the polarization process. e.g. if $p_1<p_3<p_2$ then $p_1^s<p_3^s<p_2^s$, and if $p_2=p_3<p_1$ then $p_2^s=p_3^s<p_1^s$ for all $s\in\{-,+\}^l$.
\end{mylem}
\begin{proof}
We have $\displaystyle P^-=\sum_{k=0}^4\sum_{k'=0}^4 p_kp_{k'}\mathcal{C}_{V_k\cap V_{k'}}$ and $\displaystyle P^+=\sum_{k=0}^4\sum_{k'=0}^4 p_kp_{k'}\mathcal{C}_{V_k+V_{k'}}$. Therefore, we have:
\begin{align*}
p_0^-&=p_0^2+2p_0(p_1+p_2+p_3+p_4)+2(p_1p_2+p_2p_3+p_1p_3)\\
p_1^-&=p_1^2+2p_1p_4\\
p_2^-&=p_2^2+2p_2p_4\\
p_3^-&=p_3^2+2p_3p_4\\
p_4^-&=p_4^2\\
\\
p_0^+&=p_0^2\\
p_1^+&=p_1^2+2p_1p_0\\
p_2^+&=p_2^2+2p_2p_0\\
p_3^+&=p_3^2+2p_3p_0\\
p_4^+&=p_4^2+2p_4(p_1+p_2+p_3+p_4)+2(p_1p_2+p_2p_3+p_1p_3)
\end{align*}

We can easily see that the order of $p_1^-$,$p_2^-$ and $p_3^-$ is the same as that of $p_1,p_2$ and $p_3$. This is also true for $p_1^+$,$p_2^+$ and $p_3^+$. By using a simple induction on $l$, we conclude that the order of $p_1^s,p_2^s$ and $p_3^s$ is the same as that of $p_1,p_2$ and $p_3$ for all $s\in\{-,+\}^l$.
\end{proof}

\begin{mylem}
For $k\in\{1,2,3\}$, if $\exists k'\in\{1,2,3\}\setminus\{k\}$ such that $p_k\leq p_{k'}$ then $$p_k^{(\infty)}=\lim_{l\rightarrow\infty}\frac{1}{2^l}\sum_{s\in\{-,+\}^l}p_k^s=0$$ In other words, the component of $V_k$ is killed by that of $V_{k'}$.
\end{mylem}
\begin{proof}
We know from \emph{theorem 2} that the channel $P^s$ converges almost surely to a deterministic linear channel as $l$ tends to infinity (we treat $s$ as being a uniform random variable in $\{-,+\}^l$). Therefore, the vector $(p_0^s,p_1^s,p_2^s,p_3^s,p_4^s)$ converges almost surely to one of the following vectors: $(1,0,0,0,0)$, $(0,1,0,0,0)$, $(0,0,1,0,0)$, $(0,0,0,1,0)$ or $(0,0,0,0,1)$. In particular, $p_k^s$ converges almost surely to 0 or 1.\\

Since $p_k\leq p_{k'}$ then $p_k^s\leq p_{k'}^s$ for any $s$, and so $p_k^s$ cannot converge to 1 because otherwise the limit of $p_{k'}^s$ would also be equal to 1, which is not possible since none of the 5 possible vectors contain two ones. We conclude that $p_k^s$ converges almost surely to 0, which means that $p_k^{(l)}$ (the average of $p_k^s$ on all possible $s\in\{-,+\}^l$) converges to 0. Therefore, $p_k^{(\infty)}=0$.
\end{proof}

\begin{myprop}
If $p_3\leq\max\{p_1,p_2\}$, then we have total loss in the dominant face.
\end{myprop}
\begin{proof}
If $p_3\leq\max\{p_1,p_2\}$, then by the previous lemma we have $p_3^{(\infty)}=0$. Therefore, we have total loss in the dominant face (see \emph{remark 2}).
\end{proof}

\begin{mycor}
If we do not have total loss in the dominant face then the final capacity region (to which the capacity region is converging) must be symmetric.
\end{mycor}
\begin{proof}
From the above proposition we conclude that $p_3>\max\{p_1,p_2\}$ and from lemma 9 we conclude that $p_1^{(\infty)}=p_2^{(\infty)}=0$. Thus, $I_1^{(\infty)}=I_2^{(\infty)}=p_3^{(\infty)}+p_4^{(\infty)}$ and the final capacity region is symmetric.  In particular, it contains the ``equal-rates'' rate vector.
\end{proof}

\begin{myconj}
The condition in \emph{proposition 9} is necessary for having total loss in the dominant face. i.e. if $p_3>\max\{p_1,p_2\}$, then we do not have total loss in the dominant face.
\end{myconj}

\section{Conclusion}

We have seen in this report how we can construct reliable polar codes for any $m$-user $MAC$ with inputs in $\mathbb{F}_q$. We have seen that for $0<\epsilon$ and $\beta<\frac{1}{2}$, a polar code of length $N$ can be constructed such that its sum rate is within $\epsilon$ from the sum capacity of the channel, and the probability of error is less than $2^{-N^\beta}$.\\

We have seen also that although the sum capacity is achievable with polar codes, we may lose some rate vectors from the capacity region upon polarization. We have studied this loss in the case where the channel is a combination of linear channels, and we derived a characterization of non-losing channels in this special case. We have also derived a sufficient condition for having total loss in the dominant face in the capacity region (i.e. the dominant face converges to a single point) in the case of binary input 2-user MAC.\\

Several questions are still open, the most important one is whether we can find a coding scheme, based on polar codes, in which all the symmetric capacity region is achievable.

\newpage

\bibliographystyle{abbrv}
\bibliography{PolarCodesForMAC}

\end{document}